\documentclass[10pt,a4paper]{article}
\pdfoutput=1

\usepackage{amssymb}
\usepackage{amsfonts}
\usepackage{amsmath}
\usepackage[amsmath,thmmarks]{ntheorem}
\usepackage{mathrsfs}
\usepackage{graphicx}
\usepackage{color}
\usepackage{cite}
\usepackage{subfigure}
\usepackage{textcomp}
\usepackage{algorithm}
\usepackage{algorithmic}
\usepackage{url}
\usepackage[toc,page]{appendix}

\newcommand{\R}{\mathbb{R}}

\newcommand{\id}{\textrm{id}}

\newcommand{\N}{\mathbb{N}}

\newtheorem{thm}[equation]{Theorem}
\newtheorem{lem}[equation]{Lemma}
\newtheorem{prop}[equation]{Proposition}
\newtheorem{cor}[equation]{Corollary}
\newtheorem{con}[equation]{Conjecture}

\theoremstyle{plain}
\theorembodyfont{\normalfont}
\newtheorem{defn}[equation]{Definition}
\newtheorem{ex}[equation]{Example}
\newtheorem{rem}[equation]{Remark}

\theoremstyle{nonumberplain}
\theoremheaderfont{\normalfont\bfseries}
\theorembodyfont{\normalfont}
\theoremsymbol{\ensuremath{\square}}
\theoremseparator{.}
\newtheorem{proof}{Proof}

\title{Towards a theory of statistical tree-shape analysis}

\author{Aasa~Feragen, Pechin~Lo, Marleen~de Bruijne,\\
 Mads~Nielsen, Fran\c{c}ois~Lauze\\
  E-mail: \{aasa, pechin, marleen, madsn, francois\}@diku.dk \thanks{A. Feragen, P. Lo, M. de
    Bruijne, M. Nielsen and F. Lauze are with the eScience Center, Dept.~of Computer Science, University of Copenhagen, Denmark; and M. de Bruijne
    is with the Biomedical Imaging Group Rotterdam, Depts of Radiology \& Medical
    Informatics, Erasmus MC -- University Medical Center Rotterdam, The Netherlands } }

\begin{document}

\maketitle

\begin{abstract}
In order to develop statistical methods for shapes with a tree-structure, we construct a shape space framework for tree-like shapes and study metrics on the shape space. This shape space has singularities, corresponding to topological transitions in the represented trees. We study two closely related metrics on the shape space, TED and QED. QED is a quotient Euclidean distance arising naturally from the shape space formulation, while TED is the classical tree edit distance. Using Gromov's metric geometry we gain new insight into the geometries defined by TED and QED. We show that the new metric QED has nice geometric properties which facilitate statistical analysis, such as existence and local uniqueness of geodesics and averages. TED, on the other hand, does not share the geometric advantages of QED, but has nice algorithmic properties. We provide a theoretical framework and experimental results on synthetic data trees as well as airway trees from pulmonary CT scans. This way, we effectively illustrate that our framework has both the theoretical and qualitative properties necessary to build a theory of statistical tree-shape analysis.
\end{abstract}

{\bf Keywords:} Trees, Tree metric, Shape,  Anatomical structure, Pattern matching, Pattern recognition, Geometry

\section{Introduction}

Tree-shaped objects are fundamental in nature, where they appear, e.g., as delivery systems for gases and fluids \cite{mandelbrot}, as skeletal structures, or describing
hierarchies. Examples encountered in image analysis and computational biology are airway trees~\cite{tschirren,metzen_journal,Ginneken08}, vascular
systems~\cite{Chalopin01}, shock graphs~\cite{editshock,active_skel,trinh2}, scale space hierarchies~\cite{kuijperflorack,toppoints} and phylogenetic trees \cite{phylogenetic,owen_provan,mds_phylo}. 

Statistical methods for tree-structured data would have endless applications. For instance, one could make more consistent studies of changes in airway geometry and structure related to airway disease~\cite{washko,lauge} to improve tools for computer aided diagnosis and prognosis.

Due to the wide range of applications, extensive work has been done in the past $20$ years on comparison of trees and graphs in terms of matching~\cite{keselman,metzen_journal,tschirren}, object recognition~\cite{toppoints,klein_imple} and machine
learning~\cite{ferrer_gen_med,riesen_bunke,structurespaces} based on inter-tree distances. However, the existing tree-distance frameworks are algorithmic rather than geometric. Very few attempts~\cite{wangandmarron,pcatrees,nye} have been made to build analogues of the theory for landmark point shape spaces using manifold statistics and Riemannian submersions~\cite{kendall,small,fletcher}. There exists no principled approach to studying the space of tree-structured data, and as a consequence, the standard statistical properties are not well-defined. As we shall see, difficulties appear even in the basic problem of finding the average of two tree-shapes. This paper fills the gap by introducing a shape-theoretical framework for geometric trees, which is suitable for statistical analysis.

Most statistical measurements are based on a concept of distance. The most fundamental statistic is the \emph{mean} (or prototype) $m$ for a dataset $\{x_1, \ldots, x_n\}$, which can be defined as the minimizer of the sum of squared distances to the data points:
\begin{equation} \label{meandef}
m = \textrm{argmin}_x \sum_{i=1}^n d(x, x_i)^2.
\end{equation}
This definition of the mean, called \emph{Fr\'echet mean}, assumes a space of tree-shapes endowed with a distance $d$, and is closely connected to geodesics, or shortest paths, between tree-shapes. For example, the midpoint of a geodesic from $x_1$ to $x_2$ is a mean for the two-point dataset $\{x_1, x_2\}$. Hence, if there are multiple geodesics connecting $x_1$ to $x_2$, with different midpoints, then there will also be multiple means. \emph{As a consequence, without (local, generic) uniqueness of geodesics, statistical properties are fundamentally ill-posed!}

\begin{figure} 
\centering
\subfigure[Edge matching]{\label{treetopologyviolation} \includegraphics[height=18mm]{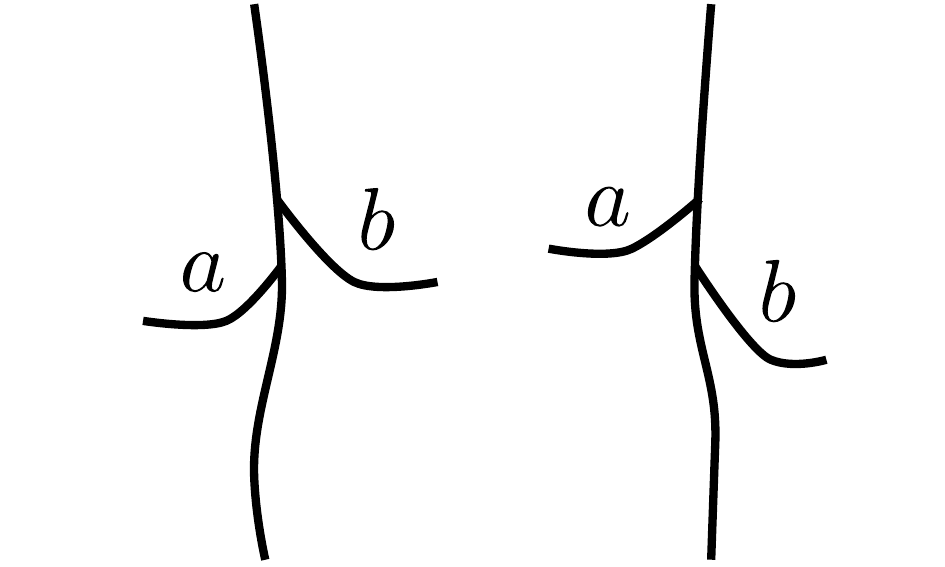}}
\hspace{2mm}
\subfigure[Geodesic candidate]{\label{treetopologyviolationgeodesic} \includegraphics[height=18mm]{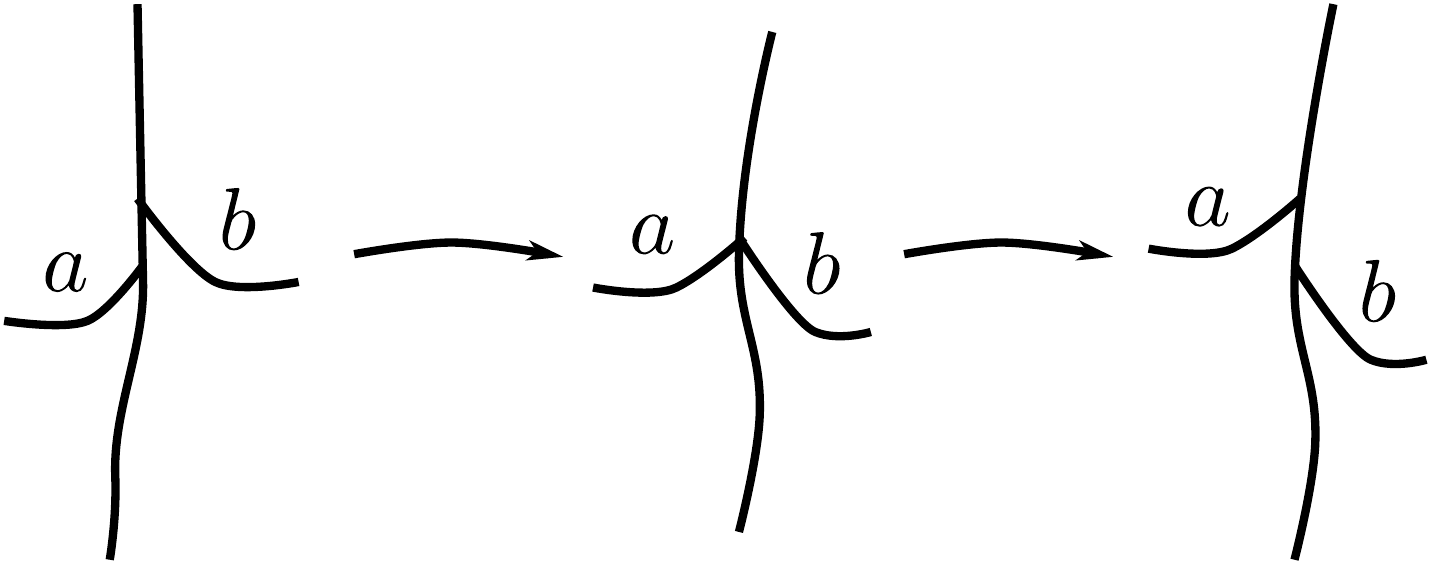}}
\caption{A good metric must handle edge matchings which are inconsistent with tree topology.}
\label{treetopology}
\end{figure}

Thus, geometry enters the picture, and the idea of a geodesic tree-space gives a constraint on the possible geometric structure of tree-space. In a shape space where distances are given by path length, we must be able to continuously deform any given tree-shape into any other by traveling along the shortest path that connects the two shapes. The deformation-paths are easy to describe when only branch shape is changed, while tree topology (branch connectivity) is fixed. Such deformations take place in portions of tree-space where all trees have the same topological structure. It is more challenging to describe deformation-paths in which the tree-topological structure is changed, for instance through a collapsed internal branch as in fig.~\ref{treetopologyviolationgeodesic}. We model topologically intermediate trees as collapsed versions of trees with differing tree topology, and glue the portions of tree-space together along subspaces that correspond to collapsed trees, as in fig.~\ref{tree_space_components}. As a consequence, tree-space has self intersections and is not smooth, but has singularities!

\begin{figure}
\centering
\includegraphics[width=\linewidth]{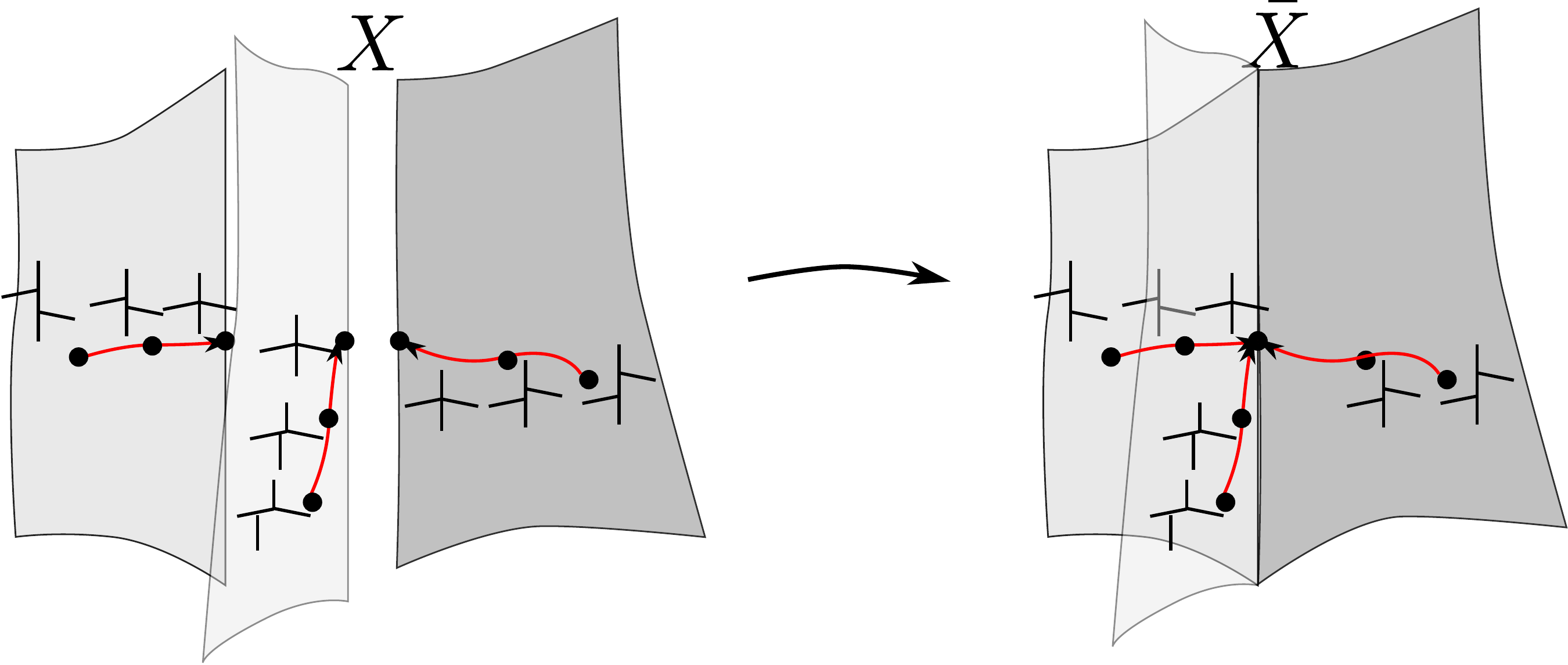}
\caption{Continuous transitions in tree topology: Tree-shapes with different tree topology live in different components of tree-space. Paths in different components can terminate in identical collapsed tree-shapes. To identify these, the tree-space components are glued together along subsets with the collapsed tree-shapes.}
\label{tree_space_components}
\end{figure}

The main theoretical contributions of this paper are the construction of a mathematical tree-shape framework along with a geometric analysis of two natural metrics on the shape space. One of these metrics is the classical Tree Edit Distance (TED), into which we gain new insight. The second metric, called Quotient Euclidean Distance (QED), is induced from a Euclidean metric. Using Gromov's approach to metric geometry \cite{gromov}, we show that QED generically gives locally unique geodesics and means, whereas finding geodesics and means for TED is ill-posed even locally. We explain why using TED for computing average trees must always be accompanied by a carefully engineered choice of edit paths in order to give well-defined results; choices which can yield average trees which are substantially different from the trees in the dataset~\cite{trinh}. The QED approach, on the contrary, allows us to investigate statistical methods for tree-like structures which have previously not been possible, like different well-defined concepts of average tree. \emph{This is our motivation for studying the QED metric!}

\begin{figure}
\centering
\includegraphics[width=\linewidth]{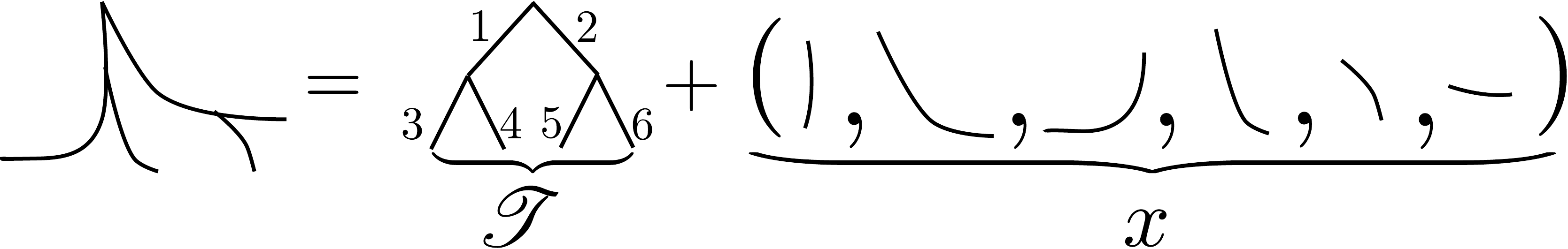}
\caption{Tree-like shapes are encoded by an ordered binary tree and a set of attributes describing edge shape.}
\label{pickapart}
\end{figure}

The paper is organized as follows: In section~\ref{background} we discuss related work. The tree-space is defined in section~\ref{geometrysection}, and the statistical properties of tree-space are analyzed in section~\ref{curvaturesection}. In section~\ref{complexitysection}, we discuss how to overcome the computational complexity of both metrics, and present a simple QED approximation. In section~\ref{experimentalsection}, we illustrate the properties of QED by computing geodesics and different types of average tree for synthetic planar data trees as well as by computing QED means for sets of $3D$ airway trees from human lungs.

\subsection{Related work} \label{background}

Metrics on sets of tree-structured data have been studied by different research communities for the past 20 years. The best-known approach is perhaps Tree Edit Distance (TED), which has been used extensively for shape matching and recognition based on medial axes and shock graphs~\cite{klein_imple,trinh,torsello}. TED and, more generally, graph edit distance (GED), are also popular in the pattern recognition community, and are still used for distance-based pattern recognition approaches to trees and graphs~\cite{ferrer_gen_med,riesen_bunke}. The TED and GED metrics will nearly always have infinitely many shortest edit paths, or geodesics, between two given trees, since edit operations can be performed in different orders and increments. As a result, even the problem of finding the average of two trees is not well posed. With no kind of uniqueness of geodesics, it becomes hard to meaningfully define and compute average shapes or modes of variation. This problem can be solved to some extent by choosing a preferred edit path~\cite{ferrer_gen_med,riesen_bunke,trinh}, but there will always be a risk that the choice has negative consequences in a given setting. Trinh and Kimia~\cite{trinh} face this problem when they use TED for computing average medial axes using the simplest possible edit paths, leading to average shapes which can be substantially different from most of the dataset shapes.

Statistics on tree-shaped objects receive growing interest in the statistical community. Wang and Marron~\cite{wangandmarron} study metric spaces of trees and define a notion of average tree called the median-mean as well as a version of PCA, which finds modes of variation in terms of \emph{tree-lines}, encoding the maximum amount of structural and attributal variation. Aydin et al.~\cite{pcatrees} extend this work by finding efficient algorithms for PCA. This is applied to analysis of brain blood vessels. The metric defined by Wang and Marron does not give a natural geodesic structure on the space of trees, as it places a large emphasis on the tree-topological structure of the trees. The metric has discontinuities in the sense that a sequence of trees with a shrinking branch will not converge to a tree that does not have that branch. Such a metric is not suitable for studying trees with continuous topological variations and noise, such as anatomical tree-structures extracted from medical images, since the emphasis on topology makes the metric sensitive to structural noise.

A different approach is that of Jain and Obermayer \cite{structurespaces}, who study metrics on attributed graphs, represented as incidence matrices. The space of graphs is defined as a quotient of the Euclidean space of incidence matrices by the group of vertex permutations. The graph-space inherits the Euclidean metric, giving it the structure of an orbifold. This graph-space construction is similar to the tree-space presented in this paper in the sense that both spaces are constructed as quotients of a Euclidean space. The graph-space framework does not, however, give continuous transitions in internal graph-topological structure, which leads to large differences between the geometries of the tree- and graph-spaces.

Trees also appear in genetics. Hillis et al.~\cite{mds_phylo} visualize large sets of phylogenetic trees using multidimensional scaling. Billera et al.~\cite{phylogenetic} have invented a phylogenetic tree-space suitable for geodesic analysis of phylogenetic trees, and Owen and Provan~\cite{owen_provan} have developed fast algorithms for computing geodesics in phylogenetic tree-space. Nye~\cite{nye} has developed a notion of PCA in phylogenetic tree-space, but is forced to make strict assumptions on possible principle components being ''simple lines'' for the sake of computability. Phylogenetic trees are not geometric, and have fixed, labeled leaf sets, making the space of phylogenetic trees much simpler than the space of tree-like shapes.

We have previously \cite{icpr,feragen_accv2010} studied geodesics between small tree-shapes in the same type of singular shape space as studied here, but most proofs have been left out. In~\cite{feragen_iccv11}, we study different algorithms for computing average trees based on the QED metric. This paper extends and continues~\cite{feragen_accv2010}, giving proofs, in-depth explanations and more extensive examples illustrating the potential of the QED metric.

\section{The space of tree-like shapes} \label{geometrysection}

Let us discuss which properties are desirable for a tree-shape model. As previously discussed, we require, at the very least, local existence and uniqueness for geodesics in order to compute average trees and analyze variation in datasets. When geodesics exist, we want the topological structure of the intermediate trees along the geodesic to reflect the resemblance in structure of the trees being compared. In particular, \emph{a geodesic passing through the trivial one-vertex tree should indicate that the trees being compared are maximally different.} Perhaps more importantly, we would like to compare trees where the desired edge matching is inconsistent with tree topology, as in fig.~\ref{treetopologyviolation}. Specifically, we would like to find geodesic deformations in which the tree topology changes when we have such edge matchings, for instance as in fig.~\ref{treetopologyviolationgeodesic}.

\subsection{Representation of trees} \label{treerep}

In this paper we shall work with two different tree-spaces: A tree-space $X$, which is the space of all trees of a certain size, and a subspace $Z \subset X$, which is a restricted space of trees, whose exact definition is flexible (see definition~\ref{subspacedef}). The large tree-space $X$ is the natural space for geometric trees. However, the available mathematical tools only allow us to prove our results locally, where the locality assumptions become very strict in $X$. Using a set of natural assumptions, we can restrict to a tree-space $Z$ where our results hold in larger regions of tree-space. This is discussed in detail in section~\ref{injectivity}. We also believe that our results hold in $X$, as described in conjecture~\ref{conjecture}.
\begin{figure}
\centering
\subfigure[]{
\label{constedge} 
\raisebox{5mm}{\includegraphics[width=0.55\linewidth]{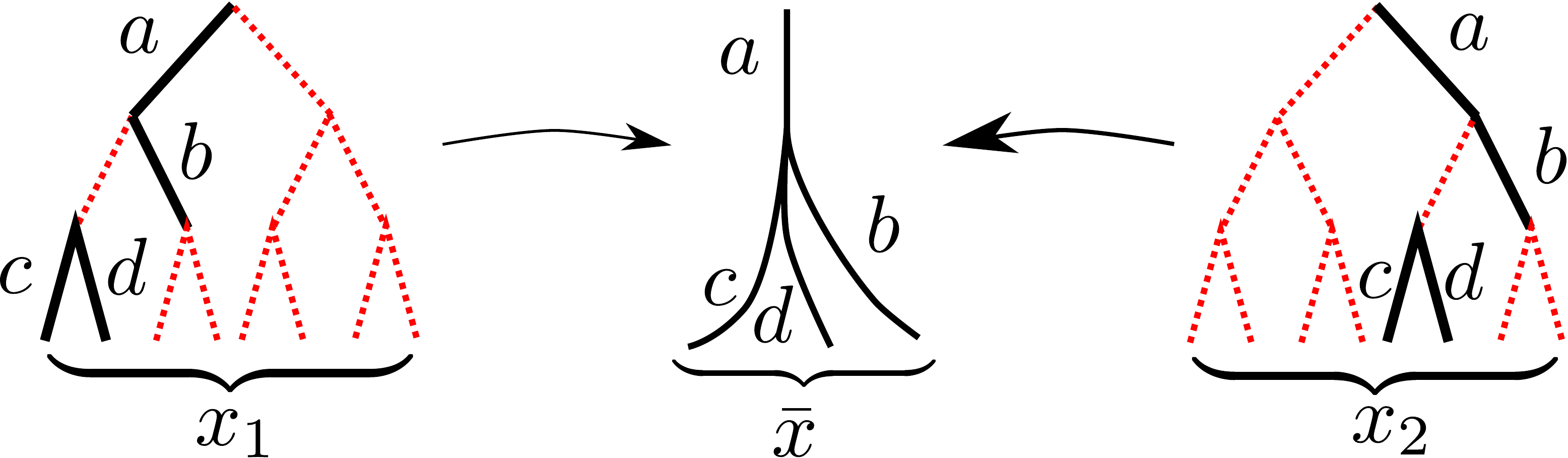}}
}
\hspace{4mm}
\subfigure[]{
\label{tedmoves4}
\includegraphics[width=0.32\linewidth]{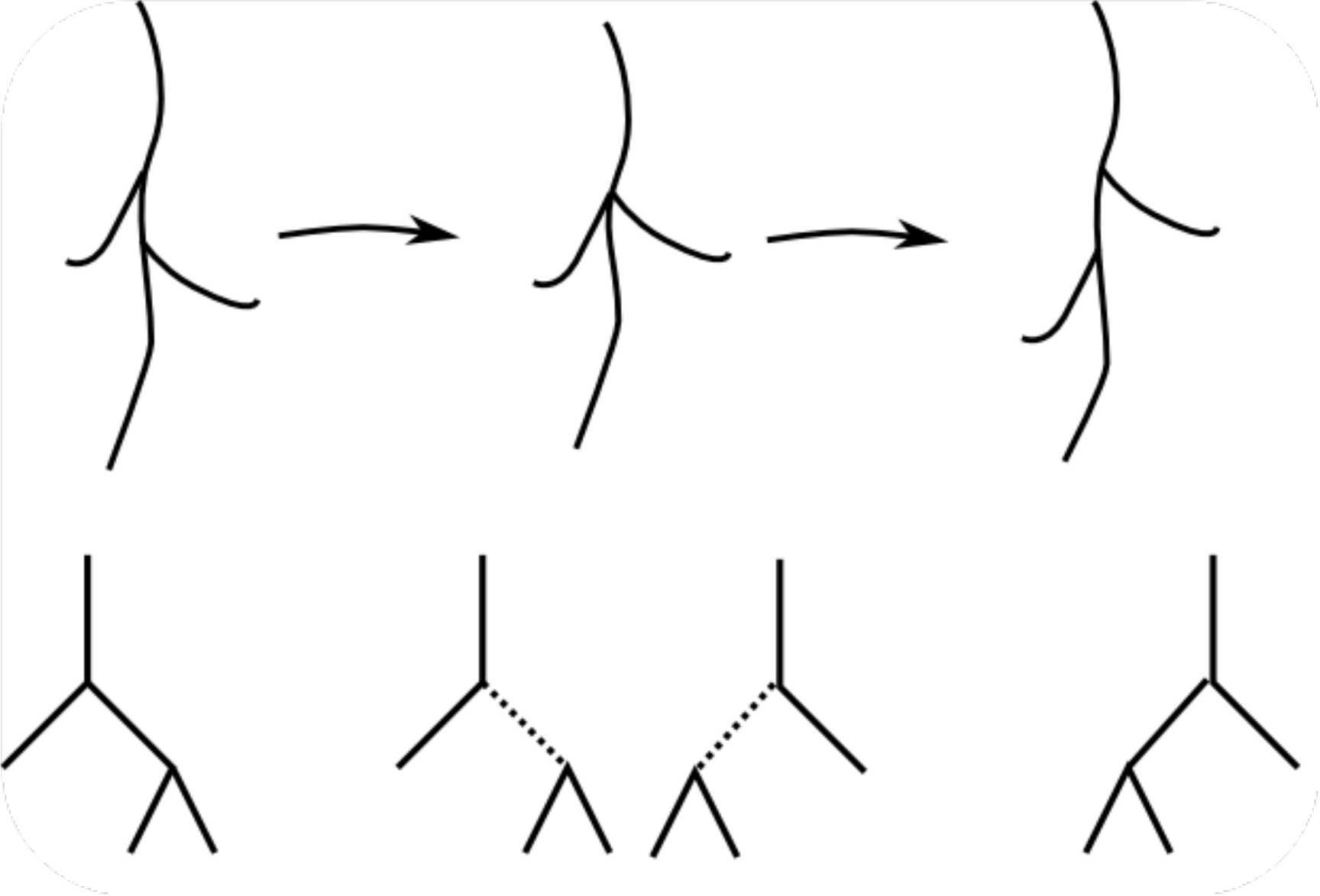}
}
\caption{(a) Higher-degree vertices are represented by collapsing internal edges (dotted line = zero attribute =  collapsed edge). We identify those tree pre-shapes whose collapsed structures are identical: The pre-shapes $x_1$ and $x_2$ represent the same tree-shape $\bar{x}$. (b) The tree deformation shown in the top row does not correspond to a path in $X$, as the two representations of the intermediate tree are found at distinct points in $X$.}
\end{figure}

In this paper, a ''tree-shape'' is an embedded tree in $\R^2$ or $\R^3$, and consists of a series of edge embeddings, glued together according to a rooted combinatorial tree. Tree-shapes are invariant to translation, but our definition of tree-shape does not remove scale and rotation. However, tree-shapes can always be aligned with respect to scale and rotation prior to comparison, if this is important.

Any tree-like shape is represented as a pair $(\mathscr{T}, x)$ consisting of a rooted, planar tree $\mathscr{T}$ with edge attributes $x$. In $\mathscr{T} = (V, E, r)$, $V$ is the vertex set, $E \subset V \times V$ is the edge set, and $r$ is the root vertex. The tree $\mathscr{T}$ describes the tree-shape topology, and the attributes describe edge shape, as illustrated in fig.~\ref{pickapart}. The shape attributes, represented by a point $x = (x_e)$ in the product space $\prod_{e \in E} A$, is a concatenation of edgewise attributes from an attribute space $A$. The attributes could, e.g., be edge length, landmark points or edge parametrizations. In this work, we mostly use open curves translated to start at the origin, described by a fixed number of landmark points along the edge. Thus, throughout the paper, the attribute space $A$ is $(\R^d)^n$ where $d = 2$ or $3$ and $n$ is the number of landmark points per edge. Collapsed edges are represented by a sequence of origin points. In some illustrations, we shall use scalar attributes for the sake of visualization, in which case $A = \R$.

In order to compare trees of different sizes and structures, we need to represent them in a unified way. We describe all shapes using the \emph{same} tree $\mathscr{T}$ to encode tree topology. By choosing a sufficiently large $\mathscr{T}$, we can represent all the trees in our dataset by filling out with empty (collapsed) edges. We call $\mathscr{T}$ \emph{the maximal tree}.

We model tree-shapes using \emph{binary} maximal trees $\mathscr{T}$. Tree-shapes which are not binary are represented by the binary tree $\mathscr{T}$ in a natural way by allowing constant, or collapsed, edges, represented by the zero scalar or vector attribute. In this way \emph{an arbitrary attributed tree can be represented as an attributed {\bf binary} tree}, see fig.~\ref{constedge}. This is geometrically very natural. Binary trees are geometrically \emph{stable} in the sense that small perturbations of a binary tree-shape do not change the tree-topological structure of the shape. Conversely, a trifurcation or higher-order vertex can always be turned into a series of bifurcations sitting close together by an arbitrarily small perturbation. In our representation, thus, trifurcations are represented as two bifurcations sitting infinitely close together, etc. 

Trees embedded in the plane have a natural edge order induced by the left-right order on the children of any edge. We say that a tree is \emph{ordered} whenever each set of sibling edges in the tree is endowed with such a total order. Conversely, an ordered combinatorial tree always has a unique, implicit embedding in the plane where siblings are ordered from left to right. For this reason we use the terms ''planar tree'' and ''ordered tree'' interchangingly. We initially study metrics on the set of ordered binary trees; later we use them to induce distances between unordered trees by considering all possible orders. This allows us to model trees in $\R^3$. Considering all orders leads to potential computational challenges, which are discussed in section~\ref{complexitysection}.

\begin{figure}
\centering
\includegraphics[width=0.8\linewidth]{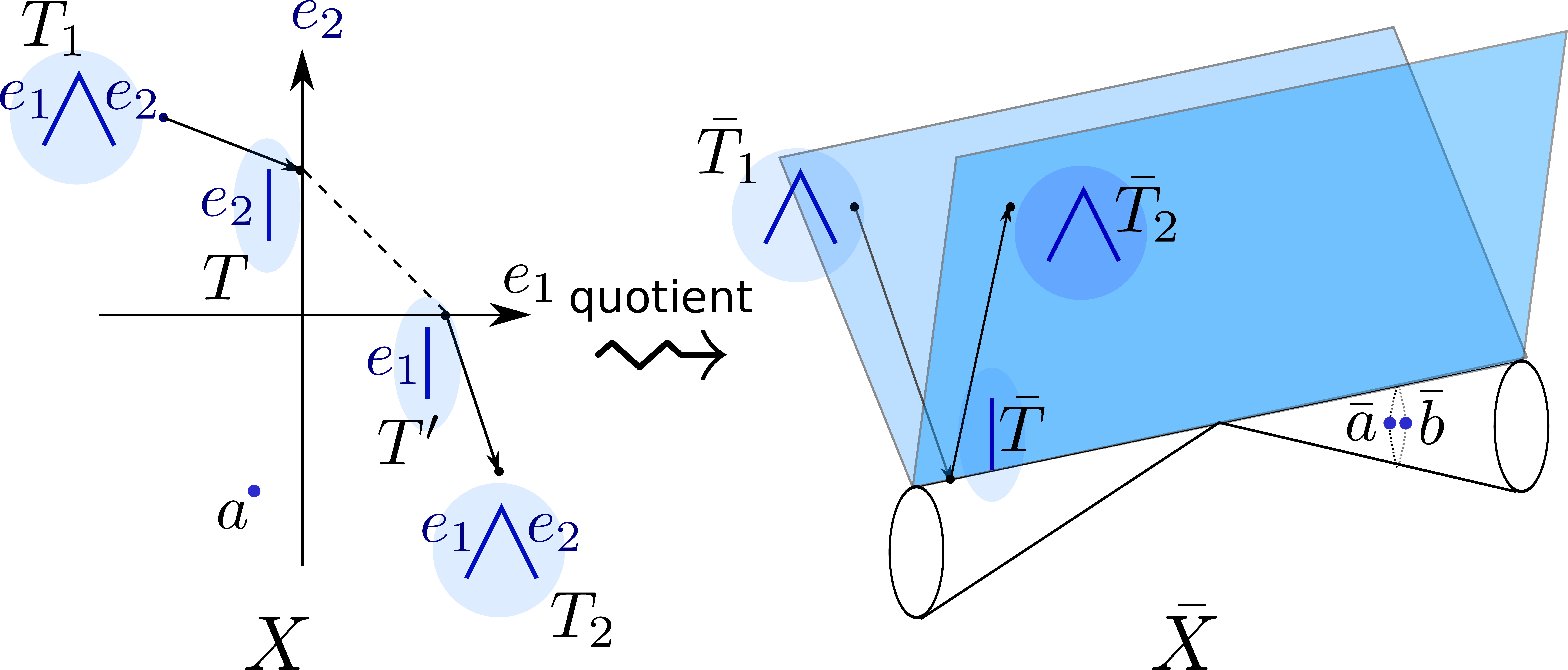}
\caption{The simplest non-trivial tree-shape space, consisting of ordered trees with two edges with scalar attributes. Along the $x$- and $y$-axes are trees with a single branch. For each real number $a$, the tree-shape found at $T' = (a,0)$ is also represented at $T = (0,a)$. We build the tree-shape space by gluing the different representations of the same tree-shapes (e.g., $T$ and $T' \in X$) together, obtaining the shape space shown on the right. Note the path from $\bar{T}_1$ to $\bar{T}_2$ through $\bar{T}$ in $\bar{X}$ on the right; the corresponding path in the pre-shape space $X$ involves a ''teleportation'' between the representations $T$ and $T'$ of $\bar{T}$.}
\label{folding}
\end{figure}

In order to build a space of tree-shapes, fix an ordered maximal binary tree $\mathscr{T}$ with edges $E$, which encodes the connectivity of all our trees. Any attributed tree $T$ is now represented by a point $x = (x_e)_{e \in E}$ in $X = \prod_{e \in E} (\R^d)^n$, where the coordinate $x_e$ describes the shape of the edge $e$. Since we allow zero-attributed edges, as discussed above, some tree-shapes will be represented by several points in $X$ (fig.~\ref{constedge}). As a result, some natural tree-deformations are not found as continuous paths in $X$. In figs.~\ref{tedmoves4} and~\ref{folding}, the paths in $X$ corresponding to the indicated deformations require a ''teleportation'' between two representations of the intermediate tree-shape. We tackle this by using a refined tree-shape space $\bar{X}$, where different representations are identified as being the same point. The original space $X$ is called the tree pre-shape space, analogous to Kendall's terminology~\cite{kendall}.

\subsection{The singular space of ordered tree-shapes} \label{singularspace}

We go from pre-shapes to shapes by identifying those pre-shapes which define the same shape. 

Consider two ordered tree-shapes with collapsed edges. Replace their binary representations by collapsed representations, where the zero attributed edges have been removed. The orders of the original trees induce well-defined orders on the collapsed trees. We say that two ordered tree-shapes are \emph{the same} when their collapsed ordered, attributed versions are identical, as in fig.~\ref{constedge}. Tree identifications come with an inherent bijection of subsets of the edge set $E$: If we identify $x, y \in X = \prod_E (\R^d)^n$, denote by
\begin{eqnarray}
E_1 = \{e \in E | x_e \neq 0\},\\
E_2 = \{e \in E | y_e \neq 0\}. 
\end{eqnarray}
the sets of non-collapsed edges with non-zero attributes. The identification of $x$ and $y$ is equivalent to an order preserving bijection $\varphi \colon E_1 \to E_2$, identifying those edges that correspond to the same edge in the collapsed tree-shape. Varying the attributes $x_e$, $\varphi$ spans a family of tree-shapes with fixed topology and several representatives. Thus, the edge sets $E_1$ and $E_2$ induce linear subspaces 
\begin{eqnarray}{l} \label{vdef}
V_1 = \{x \in X | x_e = 0 \textrm{ if } e \notin E_1\}\label{vdef1}\\
V_2 = \{x \in X | x_e = 0 \textrm{ if } e \notin E_2\}\label{vdef2}
\end{eqnarray}
of $X$ where, except for on the axes, the topological tree structure is constant. The tree-shapes represented in $V_1$ are exactly the same as those represented in $V_2$. The bijection $\varphi$ induces a bijection $\Phi \colon V_1 \to V_2$ given by $\Phi \colon (x_e) \mapsto (x_{\varphi(e)})$, which connects each representation $x \in V_1$ to the representation $\Phi(x) \in V_2$ of the same shape. Note that the $V_i$ are spanned by axes in $X$. 

Define a map $\Phi$ for each pair of identified tree-structures, and form an equivalence on $X$ by setting $x \sim \Phi(x)$ for all $x$ and $\Phi$. For each $x \in X$, $\bar{x}$ is the equivalence class $\{x' \in X| x' \sim x\}$. The quotient space 
\begin{equation} \label{qspacedef}
\bar{X} = (X/\sim) = \{\bar{x} | x \in X\}
\end{equation}
of equivalence classes $\bar{x}$ is the space of ordered tree-like shapes.

Quotient spaces are standard constructions from topology and geometry, where they are used to glue spaces together~\cite[chapter 1.5]{bridsonhaef}. The geometric interpretation of the identification in the tree-space quotient is that we fold and glue the pre-shape space space along the identified subspaces; i.e., when $x_1 \sim x_2$, we glue the two points $x_1$ and $x_2$ together. See fig.~\ref{folding} for an illustration.

\subsection{Metrics on the space of ordered trees} \label{metrics}

Given a metric $d$ on the Euclidean pre-shape-space $X = \prod_{e \in E} (\R^d)^n$, we induce
the standard quotient pseudometric \cite{bridsonhaef} $\bar{d}$ on the quotient space
$\bar{X} = X/\sim$ by setting
\begin{equation} \label{metricdefn} \bar{d}(\bar{x}, \bar{y}) = \inf \left\{ \sum_{i=1}^k
    d(x_i, y_i) | x_1 \in \bar{x}, y_i \sim x_{i+1}, y_k \in \bar{y}\right\}.
\end{equation}
This corresponds to finding the optimal path from $\bar{x}$ to $\bar{y}$, consisting of any number $k$ of concatenated Euclidean lines, passing through $k-1$ identified subspaces, as shown in fig.~\ref{folding}. It is clear from the definition that the distance function $\bar{d}$ is symmetric and transitive. It is, however, an infimum, giving a risk that the distance between two distinct tree-shapes is zero, as occurs with some intuitive shape distance functions~\cite{michormumford}. This is why $\bar{d}$ is called a \emph{pseudo}metric, and it remains to prove that it actually is a metric; i.e., that $\bar{d}(\bar{x}, \bar{y}) = 0$ implies $\bar{x} = \bar{y}$.

We prove this for two specific metrics on $X$, which come from two different ways of combining
individual edge distances: The metrics $d_1$ and $d_2$ on $X = \prod_{e \in E} (\R^d)^n$ are the norms 
\begin{eqnarray}
\|x-y\|_1 = \sum_{e \in E} \|x_e - y_e\|,\\
\|x-y\|_2 = \sqrt{\sum_{e \in E} \|x_e - y_e \|^2}.
\end{eqnarray}

From now on, $d$ and $\bar{d}$ will denote either the distance functions $d_1$ and $\bar{d}_1$, or $d_2$ and $\bar{d}_2$. We prove the following theorem in section~\ref{theoremproof}:

\begin{thm} \label{orderedthm} The distance function $\bar{d}$ is a metric
  on $\bar{X}$, which is a contractible, complete, proper geodesic space.
\end{thm}

\noindent Thus, given any two trees, we can always find a geodesic between them in both metrics $\bar{d}_1$ and $\bar{d}_2$. \footnote{It can be shown that for \emph{any} metric $d$ on $X$, the induced pseudometric $\bar{d}$ on $\bar{X}$ is a metric.}

We may often want to restrict to a subset of the large tree-space. 

\begin{defn}[Restricted tree-shape space] \label{subspacedef}
Consider a subset $Z \subset X$, which only contains all representations of trees of certain restricted topologies, defined by collapsed subtrees of the maximal tree $\mathscr{T}$. The $i^{th}$ collapsed subtree of $\mathscr{T}$ is characterized by a subset $E_i \subset E$ consisting of the edges in the maximal tree $\mathscr{T}$ which are not collapsed. Associated to it is a linear subspace $Z_i = \prod_{e \in E_i \subset E} (\R^n)^d \subset \prod_{e \in E} (\R^n)^d = X$ containing representations of all the trees of this particular topology. We include all representations of each tree topology, and obtain a restricted preshape space
\[
Z = \bigcup_i Z_i \subset \prod_{e \in E} (\R^n)^d = X,
\]
containing all the trees that have of one of the considered topologies. The equivalence relation $\sim$ on $X$ restricts to an equivalence relation $\sim_Z$ on $Z$, from which we obtain a restricted tree-shape space $\bar{Z} = Z/\sim_Z = p(Z) \subset \bar{X}$. The metric $d$ on $X$ induces a metric $d_Z$ on $Z$ which induces a quotient pseudometric $\bar{d}_Z$ on $\bar{Z}$.
\end{defn}

\begin{ex}
Denote by $Z$ the space of all trees in $X$ with $n$ leaves, now $\bar{Z}$ is the space of tree-shapes with $n$ leaves.
\end{ex}

\begin{rem}
Note that the space of tree-shapes on $n$ leaves in \emph{i)} is different from the Billera-Holmes-Vogtmann (BHV) space \cite{phylogenetic} of trees with vector attributes, because in the BHV space, geodesics will always deform leaves onto leaves, whereas in $\bar{Z}$, leaves can be transformed to non-leaf branches by a geodesic.
\end{rem}

Even in the restricted tree-shape space, we obtain an induced metric:

\begin{cor} \label{subspacemetric}
The pseudometric $\bar{d}_{\bar{Z}}$ on $\bar{Z}$ is a metric, and $(\bar{Z}, \bar{d}_{\bar{Z}})$ a contractible, complete, proper metric space.
\end{cor}

\begin{proof}
First, we show that the pseudometric is a metric. The pseudometric in equation~\eqref{metricdefn} defines the distance $\bar{d}(\bar{x}, \bar{y})$ as the infimum of lengths of paths in $\bar{X}$ connecting $\bar{x}$ and $\bar{y}$. Any path in $\bar{Z}$ is also a path in $\bar{X}$, so if $\bar{d}_{\bar{Z}}(\bar{x}, \bar{y}) = 0$, then $\bar{d}(\bar{x}, \bar{y}) = 0$ as well, so $\bar{x} = \bar{y}$.

The proofs of the other claims follow the proof of theorem~\ref{orderedthm} as in section~\ref{theoremproof}.
\end{proof}

\subsection{From ordered to unordered trees}

The world is not two-dimensional, and for most applications it is necessary to study embedded trees in $\R^3$. The main difference from the ordered case is that trees in $\R^3$ have no canonical edge order. The left-right order on children of planar trees gives an implicit preference for edge matchings, and hence reduces the number of possible matches. When we no longer have this preference, we consider all orderings of the same tree and choose orders which minimize the distance.

We define the space of (unordered) tree-like shapes in $3D$ as the quotient $\bar{\bar{X}} = \bar{X} /G$, where $G$ is the group of reorderings of the maximal binary tree $\mathscr{T}$. The metric $\bar{d}$ on $\bar{X}$ induces a quotient pseudometric $\bar{\bar{d}}$ on $\bar{\bar{X}}$. Again, we can prove:

\begin{thm} \label{unorderedthm}
For $\bar{\bar{d}}$ induced by either $\bar{d}_1$ or $\bar{d}_2$, the function $\bar{\bar{d}}$ is a metric and the space $(\bar{\bar{X}}, \bar{\bar{d}})$ is a contractible, complete, proper geodesic space.

The same result holds for restricted tree-spaces with several restricted tree topologies: Let $Z \subset X$ be a subspace containing only trees of certain restricted topologies, as in definition~\ref{subspacedef}, which is saturated with respect to the reordering group\footnote{$Z$ is saturated if, for each tree topology appearing in $Z$, all reorderings of the same tree topology also appear in $Z$. Equivalently, $Z = p^{-1}(p(Z))$, where $p \colon \bar{X} \to \bar{X}/G$} $G$. For the corresponding restricted tree-space $\bar{\bar{Z}} = \bar{Z}/G \subset \bar{\bar{X}}$, the quotient pseudometric $\bar{\bar{d}}_Z$ is a metric and the space $(\bar{\bar{Z}}, \bar{\bar{d}}_Z)$ is a contractible, complete, proper geodesic space. \hfill \qed
\end{thm}

%

Note that $G$ is a finite group, which means that $\bar{\bar{X}}$ is locally well-behaved almost everywhere. In particular, off fixed-points for the action of reorderings $g \in G$ on $\bar{X}$, the projection $\bar{p} \colon \bar{X} \to \bar{\bar{X}}$ is a local isometry, i.e., it is distance preserving within a neighborhood. Hence, the geometry from $\bar{X}$ is preserved off the fixed points. Geometrically, a fixed point in $\bar{X}$ is an ordered tree-shape where a reordering $g$ of certain branches does not change the ordered tree-shape; that is, some pair of sibling edges must have the same shape attributes. In particular, the fixed points are non-generic because they belong to the lower-dimensional subset of $\bar{X}$ where two sibling edges have identical shape. Theorem~\ref{unorderedthm} can be proved using standard results on compact transformation groups along with similar techniques as for theorem~\ref{orderedthm}.

While considering all different possible orderings of the tree is easy from the point of view of geometric analysis, in reality this becomes a computationally impossible task when the size of the trees grow beyond a few generations. In real applications we can, however, efficiently reduce complexity using heuristics and approximations, as discussed in section~\ref{complexitysection}.

\subsection{Geometric interpretation of the metrics} \label{interpretation}

It follows from the definition that the metrics $\bar{d}_1$ and $\bar{\bar{d}}_1$ coincide with the classical tree edit distance (TED) metric for ordered and unordered trees, respectively. In this way the abstract, geometric construction of tree-space gives a new way of viewing the intuitive TED algorithm.

The metrics $\bar{d}_2$ and $\bar{\bar{d}}_2$ are descents of the Euclidean metric on $\bar{X}$, and geodesics in this metric are concatenations of straight lines in flat regions. In section~\ref{fundamental} we compare the two metrics using examples.

\subsection{Proof of theorem~\ref{orderedthm}} \label{theoremproof}

We now pass to the proof of theorem~\ref{orderedthm}. The rest of the article is independent of this section, and during the first read, the impatient reader may skip to section~\ref{curvaturesection}. However, while the proof is technical, theorem~\ref{orderedthm} is a fundamental building block for the shape space framework. We shall assume that the reader has a good knowledge of metric geometry or general topology \cite{bridsonhaef,dug}. It is crucial for the proof that we are only identifying subspaces of the Euclidean space $X$ which are spanned by Euclidean axes, and these are finite in number. This induces a well-behaved projection $p \colon X \to \bar{X}$, which carries many properties from $X$ to $\bar{X}$.

\subsubsection{Precise shape space definition} 

We say that ordered tree-shapes whose (collapsed) ordered structure is the same, belong to the same \emph{combinatorial tree-shape type}. For each combinatorial type of ordered tree-shape $C_j$ ($j = 1, \ldots, K$) which can be represented by collapsing edges in the maximal tree $\mathscr{T} = (V, E, r, <)$, there is a family $E^i_j$ of subsets of $E$, which induce that particular type $C_j$ when we endow the edges in $E^i_j$ ($i = 1, \ldots, n_j$) with nonzero attributes and leave all other edges with zero attributes. These subsets are characterized by the properties
\begin{itemize}
\item[a)] the cardinality $|E^1_j| = | E^i_j|$ for all $i=1, \ldots, n_j$, $j = 1, \ldots, K$.
\item[b)] there is a depth-first order on each $E^i_j$ induced by the depth-first order on $E$, such that the ordered, combinatorial structure defined by any $E^i_j$ coincides with that defined by $E^1_j$.
\end{itemize}

That is, the subset $E^i_j$ for any $i$ lists the set of edges in $\mathscr{T}$ which have nonzero attributes for the $i^{th}$ representation of any shape of type $C_j$. Corresponding to each $E^i_j$ is the linear subspace $V^i_j$ of $X$ given by 
\begin{equation}
V^i_j = \left\{(x_e) \in \prod_{e \in E} (\R^d)^n | x_e = 0 \textrm{ if } e \notin E^i_j\right\}
\end{equation}
and by condition b) we can define isometries $\phi^i_j \colon V^i_j \to \prod_{e \in E^1_j} \R^{dn}$ by forgetting the zero entries in $V^i_j$ and keeping the depth-first coordinate order. We generate the equivalence $\sim$ on $X$ by asking that $z \sim w$ whenever $\phi^i_j(z) = \phi^l_j(w)$ for some $i,j,l$. We now define the space of ordered tree-like shapes as the quotient $\bar{X} = X/\sim$, and define the quotient map $p \colon X \to \bar{X}$.

\subsubsection{The pseudometric is a metric}

It is clear from the definition that the distance function $\bar{d}$ defined in \eqref{metricdefn} is symmetric and satisfies the triangle inequality, which makes it a \emph{pseudometric}.

\begin{prop} \label{pseudometricismetric}
Let $d$ denote $d_1$ or $d_2$. The pseudometric $\bar{d}$ is a metric on $\bar{X}$.
\end{prop}

\begin{proof}
It suffices to show that $\bar{d}_i(\bar{x}, \bar{y}) = 0$ implies $\bar{x} = \bar{y}$ ($i = 1, 2$). Moreover, it is also easy to show that $d_1(\bar{x}, \bar{y}) \ge d_2(\bar{x}, \bar{y})$ for any $\bar{x}, \bar{y} \in \bar{X}$, so it suffices to show that $\bar{d}_2(\bar{x}, \bar{y}) = 0$ implies $\bar{x} = \bar{y}$. Hence, from now on, write $\bar{d}$ for $\bar{d}_2$, and assume that $\bar{d}(\bar{x}, \bar{y}) = 0$ for two tree-shapes $\bar{x}$ and $\bar{y}$.

Choose $\epsilon > 0$ such that 
\begin{equation}
\epsilon \ll \min \left\{
\begin{array}{c|c}
\begin{array}{c}
\|x_e\| > 0,\\
\|y_e\| > 0,\\
\|x_e - x_{\tilde{e}}\| > 0,\\
\|y_e - y_{\tilde{e}} \| > 0,\\
\end{array}  
&
\begin{array}{c}
x = (x_e) \in \bar{x},\\
y = (y_e) \in \bar{y}
\end{array}
\end{array} \right\},
\end{equation}  
that is, $\epsilon$ is smaller than the size of any of the non-zero edges in $\bar{x}$ and $\bar{y}$.

We may assume that $x, y \in \bigcup_{i,j} V^i_j$ since otherwise we may assume by symmetry that $\bar{x} = \{x\}$ and $\bar{d}(\bar{x}, \bar{y}) \ge \min \{d(x, \bar{y}), d(x, \bigcup_{i,j} V^i_j)\} > 0$.

Denote by $\bar{X}_j$ the image of $V^i_j$ under the quotient projection $p \colon X \to \bar{X}$ for any $i$.

We may assume that $\bar{x}$ and $\bar{y}$ belong to the same identified subspace; that is, there exist $i, j$ such that 
\begin{equation} \label{inthesame}
\bar{x} \cap V^i_j \neq \emptyset, \ \bar{y} \cap V^i_j  \neq \emptyset
\end{equation}
since otherwise, 
\begin{equation} \label{emptyint}
\bar{y} \cap \left( \bigcup \{V^i_j | \bar{x} \cap V_j^i \neq \emptyset\} \right) = \emptyset \textrm{ for all } i, j.
\end{equation}
Since $\bar{y}$ is a finite set, and $\bigcup \{V^i_j | \bar{x} \cap V^i_j \neq \emptyset \}$ is a closed set, (\ref{emptyint}) implies
\begin{equation}
d\left(\bar{y}, \bigcup \{V^i_j | \bar{x} \cap V^i_j \neq \emptyset \} \right) > 0.
\end{equation}
In this case, the path will have to go through some $V^{\tilde{i}}_{\tilde{j}}$ which does not contain points equivalent to $y$, and
\begin{equation}
d \left( \bar{y}, \bigcup \{ V^i_j | \bar{y} \cap V^i_j = \emptyset \} \right) > \epsilon,
\end{equation}
since in order to reach $\bigcup \{ V^i_j | \bar{y} \cap V^i_j = \emptyset \}$, we need to remove edge attributes which are nonzero in $\bar{y}$, and $\epsilon \ll \|y_e\|$ for all $y_e \neq 0$. Thus, eq.~\ref{inthesame} holds, and in fact, it holds for all the intermediate path points $\bar{x}_i$ from eq.~\ref{metricdefn}.

But if the path points stay in $\bar{X}_j$, then the path consists of shifting and changing the nonzero edge attributes of the trees in question. This will only give a sum $< \epsilon$ if the trees are identical and the path is constant.
\end{proof}

\subsubsection{Topology of the space of tree-like shapes}

Here, we prove the rest of theorem~\ref{orderedthm}, namely that the tree-shape space $(\bar{X}, \bar{d})$ is a complete, proper geodesic space, and $\bar{X}$ is contractible. First, we note that although $\bar{X}$ is not a vector space, there is a well-defined notion of size for elements of $\bar{X}$, induced by the norm on $X$:

\begin{lem} \label{sizelemma}
Note that if $x \sim y$, we must have $\| x \| = \| y \|$; hence we can define $\| \bar{x} \| :=\| x \|$.
\end{lem}

\begin{proof}
The equivalence is generated recursively from the conditions $x \sim y$ whenever either $x = y$, indicating $\|x\| = \|y\|$; or $\phi^i_j(x) = \phi^i_k(y)$, indicating $\|x\|=\|\phi^i_j(x)\|=\|\phi^i_k(y)\|=\|y\|$ since the $\phi$ are isometries. Hence, the lemma holds by recursion.
\end{proof}

We will prove that $(\bar{X}, \bar{d})$ is a proper geodesic space using the Hopf-Rinow theorem for metric spaces \cite{bridsonhaef}, which states that every complete locally compact length space is a proper geodesic space. A length space is a metric space in which the distance between two points can always be realized as the infimum of lengths of paths joining the two points. Note that this is a weaker property than being a geodesic space, as the geodesic joining two points does not have to exist; it is enough to have paths that are arbitrarily close to being a geodesic. It follows from \cite[chapter I lemma 5.20]{bridsonhaef} that $(\bar{X}, \bar{d})$ is a length space for \emph{any} metric $d$ on $X$ where $\bar{d}$ is a metric.

To see that the tree-shape space is locally compact, note that the projection $p \colon X \to \bar{X}$ is finite-to-one, so any open subset $U$ of $\bar{X}$ has as pre-image a finite union $\bigcup_{i=1}^N U_i$ of open subsets of $X$, such that $\textrm{diam}(U_i) = \textrm{diam}(U)$ and $p(\bigcup_i \bar{U}_i) = \bar{U}$ is compact whenever $U$ is bounded.

We also need to prove that $(\bar{X}, \bar{d})$ is complete:

\begin{prop} \label{completenessdescends}
Let $\bar{d}$ denote either of the metrics $\bar{d}_1$ and $\bar{d}_2$. The shape space $(\bar{X}, \bar{d})$ is complete.
\end{prop}

The proof needs a lemma from general topology:

\begin{lem} \emph{\cite[chapter XIV, theorem~2.3]{dug}} \label{duglemma}
Let $(Y,d)$ be a metric space and assume that the metric $d$ has the following property: \emph{There exists $\epsilon > 0$ such that for all $y \in Y$ the closed ball $\bar{B}(y, \epsilon)$ is compact.} Then $\bar{d}$ is complete. 
\hfill \qed
\end{lem}

Using the projection $p \colon X \to \bar{X}$, we can prove:

\begin{lem} \label{closedandbounded}
Bounded closed subsets of $\bar{X}$ are compact.
\end{lem} 

\begin{proof}
Since lemma~\ref{sizelemma} defines a notion of size in $\bar{X}$, any closed, bounded subspace $C$ in $\bar{X}$ is contained in a closed ball $\bar{B}_{\bar{d}}(\bar{0}, R)$ in $\bar{X}$ for some $R > 0$, where $\bar{0} $ is the image $p(0) \in \bar{X}$. Since $\|x\| = \| \bar{x} \|$, it follows that $p^{-1}(\bar{B}_{\bar{d}}(\bar{0}, R)) = \bar{B}_d(0, R)$, which is a closed and bounded ball in $X$. Since closed, bounded subsets of $X$ are compact $\bar{B}_d(0, R)$ is compact. By continuity of $p$, $\bar{B}_{\bar{d}}(\bar{0}, R))$ is compact. Then $C$ is compact too.
\end{proof}

It is now very easy to prove proposition~\ref{completenessdescends}:

\begin{proof}[Proof of proposition~\ref{completenessdescends}]
By lemma~\ref{closedandbounded}, all closed and bounded subsets of $\bar{X}$ are compact, but then by lemma~\ref{duglemma} the metric $\bar{d}$ must be complete.
\end{proof}

Using the Hopf-Rinow theorem \cite[chapter I, proposition~3.7]{bridsonhaef} we thus prove that $(\bar{X}, \bar{d})$ is a complete, proper geodesic space. We still miss contractibility:

\begin{lem} \label{contractible}
Let $B$ be a normed vector space and let $\sim$ be an equivalence on $B$ such that $a \sim b$ implies $t \cdot a \sim t \cdot b$ for all $t \in \R$. Then $\bar{B} = B/\sim$ is contractible.
\end{lem}

\begin{proof}
Define a map $H \colon \bar{B} \times [0,1] \to \bar{B}$ by setting $H(\bar{x}, t) = t \cdot \bar{x}$. Now $H$ is well defined because of the condition on $\sim$, and $H(\bar{x}, 0) = 0 \ \forall \ \bar{x} \in \bar{B}$ so $H$ is a homotopy from $\id_{\bar{B}}$ to the constant zero map.
\end{proof}

Combining the results of section~\ref{theoremproof}, we see that the proof of theorem~\ref{orderedthm} is complete. 

\section{Curvature in the tree-shape space} \label{curvaturesection}

Having proved theorem~\ref{orderedthm}, we may now pass to studying the geometry of the tree-shape space through its geodesics. Uniqueness of geodesics and means is closely connected to the geometric notion of \emph{curvature}, a concept which fundamentally depends on the underlying metric. Using methods from metric geometry \cite{bridsonhaef,gromov} we shall investigate the curvature of the tree-shape space with the QED and TED metrics. Using curvature, we obtain well-posed statistical methods for QED.

The next theorem states that in the tree-shape space endowed with the QED metric, any randomly selected point has a corresponding neighborhood within which the tree-space has non-positive curvature. We shall use this fact to show that datasets within that same neighborhood have unique averages.

\begin{thm}  \label{qedcurv}
\begin{itemize}
\item[i)] Endow $\bar{X}$ with the QED metric $\bar{d}_2$. A generic point $\bar{x} \in \bar{X}$ has a neighborhood $U \subset \bar{X}$ in which the curvature is non-positive. At non-generic points, the curvature of $(\bar{X}, \bar{d}_2)$ is $+\infty$, or unbounded from above. 
\item[ii)] Endow $\bar{X}$ with the TED metric $\bar{d}_1$. The metric space $(\bar{X}, \bar{d}_1)$ does not have locally unique geodesics anywhere, and the curvature of $(\bar{X}, \bar{d}_1)$ is $+\infty$ everywhere.
\end{itemize}

Claims i) and ii) also hold in the subspace $\bar{Z} \subset \bar{X}$ containing only trees of certain restricted topologies, as defined in definition~\ref{subspacedef}.
\end{thm}
Here, local uniqueness is defined as uniqueness within a sufficiently small neighborhood.

\subsection{Genericity}

\emph{Genericity} is a key concept in this paper. Many of our results, e.g., uniqueness of means, do not hold in general, but they do hold for a randomly chosen dataset with respect to natural probability measures.

\begin{defn}[Generic property]
A \emph{generic property} in a metric space $(X, d)$ is a property which holds on an open, dense subset of $X$. 
\end{defn}

In the tree-spaces $\bar{X}, \bar{\bar{X}}, \bar{Z}$ and $\bar{\bar{Z}}$, one interpretation is that generic properties hold almost surely, or with probability one, with respect to natural probability measures. Thus, for a random tree-shape, we can safely assume that it satisfies generic properties, e.g., that the tree-shape is a binary tree. A \emph{non-generic property} is a property whose ''not happening'' is generic. This is similarly interpreted as a property that may \emph{not} hold for randomly selected tree-shapes. A detailed discussion of the relation between genericity and probability is found in Appendix A.

One common misconception is that the term ''generic tree'' refers to a particular class of trees with a particular generic property. It is important to note that many different properties, which do not necessarily happen on the same subsets of tree-space, may all be generic at the same time. However, any finite set of generic properties \emph{will} all happen on an open, dense subset.

\begin{prop} \label{generic}
Tree-shapes that are truly binary (i.e.~their internal edges are not collapsed) are generic in the space of all tree-like shapes. 
\end{prop}

\begin{proof}
Let $\tilde{T}$ be a tree-shape in $\bar{X}$ or $\bar{\bar{X}}$ which is not truly binary, represented by a maximal binary tree $\mathscr{T}$. By adding arbitrarily small noise to the zero attributes on edges of $\mathscr{T}$, we obtain truly binary tree-shapes $\tilde{T}'$ which are arbitrarily close to $\tilde{T}$. Thus, the set of full truly binary tree-shapes in $\bar{X}$ or $\bar{\bar{X}}$ is open and dense. Hence, truly binary tree-shapes are generic both in $\bar{X}$ and in $\bar{\bar{X}}$.
\end{proof}

\begin{figure}
\centering
\includegraphics[width=0.9\linewidth]{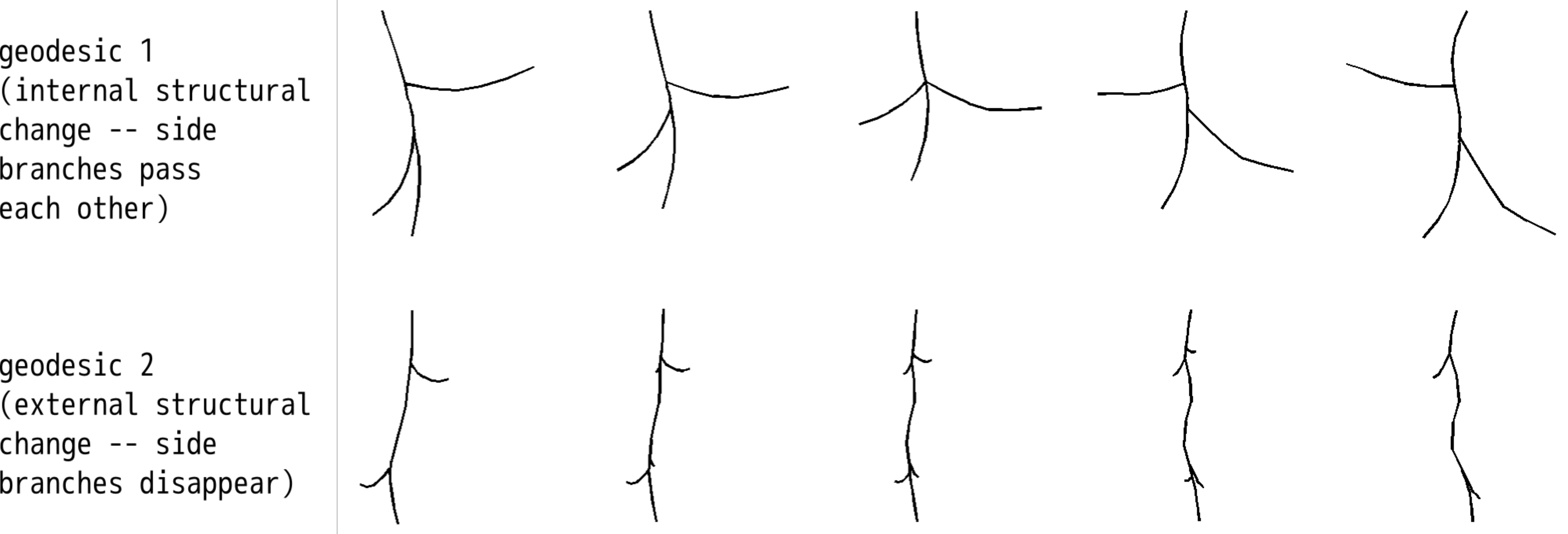}
\caption{Examples of geodesic deformations: geodesic $1$ goes through a tree with a trifurcation, while geodesic $2$ does not have internal structural transitions.}
\label{fundamentalfigure}
\end{figure}

The essence of the proposition is that binary tree-shapes are generic, but that does not mean that non-binary trees do not need to be considered! While non-binary trees may not appear as randomly selected trees, they \emph{do} appear in geodesics between randomly selected pairs of trees, as in fig.~\ref{fundamentalfigure}. Non-binary tree-like shapes also appear as samples in real-life applications, e.g., when studying airway trees. However, we interpret this as an artifact of resolution rather than true higher-degree vertices. For instance, airway extraction algorithms record trifurcations when the lengths of internal edges are below certain threshold values.

\subsection{Curvature in metric spaces}

In order to understand and prove theorem~\ref{qedcurv}, we need a definition of curvature in metric spaces. In spite of its simplicity and elegance, this concept from metric geometry is novel in computer vision. We shall spend a little time introducing it before proceeding to prove theorem~\ref{qedcurv} in section~\ref{qedcurvproof}.

Since general metric spaces can have all kinds of anomalies, the concept of curvature in such spaces is defined through a comparison with spaces that are well understood. More precisely, the metric spaces are studied using \emph{geodesic triangles}, which are compared with corresponding \emph{comparison triangles} in model spaces with a fixed curvature $\kappa$. The model spaces are spheres ($\kappa > 0$), the plane $\R^2$ ($\kappa = 0$) and hyperbolic spaces ($\kappa < 0$), and through comparison with these spaces, we can bound the curvature of the metric space by $\kappa$. In this paper we shall use comparison with planar triangles, which gives us curvature bounded from above by $0$.

\begin{figure} 
\centering
\includegraphics[width=0.5\linewidth]{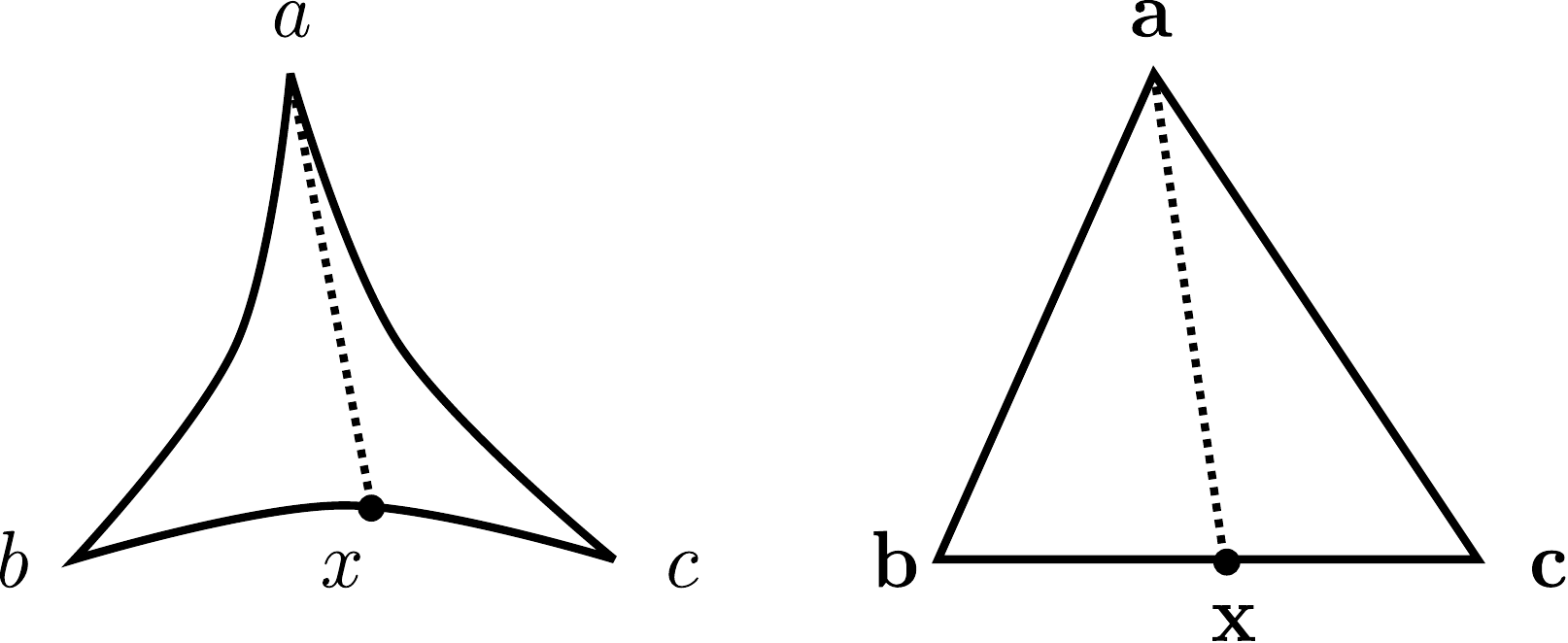}
\caption{A metric space is a $CAT(0)$ space if, for a geodesic triangle $abc$ and for any point $x$ on the triangle, the distance from $x$ to the opposite vertex is not longer than the corresponding distance in the planar comparison triangle $\mathbf{a} \mathbf{b} \mathbf{c}$.}
\label{cat0}
\end{figure}

Given a geodesic metric space $Y$, a \emph{geodesic triangle} $abc$ in $Y$ consists of three points $a, b, c$ and geodesic segments joining them. A planar \emph{comparison triangle} $\mathbf{a}\mathbf{b}\mathbf{c}$ for the triangle $abc$ consists of three points $\mathbf{a}, \mathbf{b}, \mathbf{c}$ in the plane, such that the lengths of the sides in $\mathbf{a} \mathbf{b} \mathbf{c}$ are the same as the lengths of the sides in $abc$, see fig.~\ref{cat0}.

A $CAT(0)$ space is a metric space in which geodesic triangles are ''thinner'' than for their comparison triangles in the plane. That is, $d(x, a) \le \|\mathbf{x} - \mathbf{a}\|$ for any $x$ on the edge $bc$ where $\mathbf{x}$ is the unique point on the edge $\mathbf{b} \mathbf{c}$ such that $d(b,x) = \|\mathbf{b} - \mathbf{x}\|$ and $d(x,c) = \|\mathbf{x} - \mathbf{c}\|$. If the planar comparison triangle is replaced by a comparison triangle in the sphere or hyperbolic space of fixed curvature $\kappa$, we get a $CAT(\kappa)$ space.

A space $Y$ has non-positive curvature if it is locally $CAT(0)$, i.e., if any point $x \in Y$ has a radius $r$ such that the ball $B(x, r)$ is $CAT(0)$. Similarly, define curvature bounded by $\kappa$ as being locally $CAT(\kappa)$.

\begin{ex} \label{negcurvex}
\begin{itemize}
\item[a)] The space $U$ obtained by gluing a family Euclidean spaces $U_i$ together along isomorphic affine subspaces $V_i \subset U_i$ is a $CAT(0)$ space. At any point in $U$ which is not a glued point, the local curvature is $0$, since the space is locally isomorphic to the corresponding $U_i$. At any glued point, it can be shown that the local curvature is $-\infty$. 
\item[b)] The GPCA construction by Vidal et al.~\cite{vidal} defines a $CAT(0)$ space, giving a potential use of $CAT(0)$ spaces and metric geometry in machine learning.
\item[c)] The space of phylogenetic trees is a $CAT(0)$ space~\cite{phylogenetic}.
\item[d)] As we are about to see, the space of tree-like shapes is locally a $CAT(0)$ space almost everywhere.
\end{itemize}
\end{ex}

One of the main reasons why $CAT(\kappa)$ spaces (and $CAT(0)$ in particular) are attractive, is due to the following result on existence and uniqueness of geodesics.
%
%

\begin{prop}\emph{\cite[Proposition II 1.4]{bridsonhaef}} \label{bhprop}
Let $(Y, d)$ be a $CAT(\kappa)$ space. If $\kappa \le 0$, then all pairs of points have a unique geodesic joining them. For $\kappa > 0$, the same holds for pairs of points at a distance less than $\pi/\sqrt{\kappa}$. \hfill \qed
\end{prop}

More results on curvature in metric spaces can be found in the book by Bridson and Haefliger \cite{bridsonhaef}.

\subsection{Curvature in the space of ordered tree-shapes -- proof of theorem~\ref{qedcurv}} \label{qedcurvproof}

In this section we study the curvature of tree-shape space using the theory of $CAT(0)$ spaces. We show that at a generic tree, the shape space has bounded curvature. The results rest on the following theorem:

\begin{thm} \label{unique_geo}
At a generic point $\bar{x} \in \bar{X}$, the shape space is locally $CAT(0)$, and thus, $\bar{X}$ has locally unique geodesics in a neighborhood of $\bar{x}$.
\end{thm}

\begin{proof}
Recall from section~\ref{singularspace} how $\bar{X}$ was formed by identifying subspaces $V_i \subset X$, defined in eq.~\ref{vdef1} and~\ref{vdef2}. These identified subspaces corresponded to different representations in $X$ of the same shape $\bar{x} \in \bar{X}$. The points in $\bar{X}$ can now be divided into three categories: 
\begin{itemize}
\item[i)] points $\bar{x}_1 \in \bar{X}$ which do not belong to the image of an identified subspace because they only have one representative in $X$,
\end{itemize}
and two classes of points in $\bar{X}$ which have more than one representative in $\bar{X}$.
\begin{itemize}
\item[ii)] The first class contains points $\bar{x}_2 \in \bar{X}$ at which the space $\bar{X}$ is locally homeomorphic to an intersection of linear spaces, as in example~\ref{negcurvex} a).  These points are images of points $x_2 \in X$ which belong to one single identified subspace $V_i$.
\item[iii)] The second class contains points $\bar{x}_3 \in \bar{X}$ whose preimages $x_3$ in $X$ are at the intersection of identified subspaces $V_i, V_j \subset X$. An example of such points is the image of the origin in fig.~\ref{folding}. These points correspond to trees where, infinitely close to the same tree, we can find pairs of trees in $\bar{X}$ between which geodesics are not unique.
\end{itemize}
These three classes of points correspond to local curvature $0$, $-\infty$ and $+\infty$. That is, the space is locally $CAT(0)$ at points in categories \emph{i)}; at points from \emph{ii)} it is $CAT(\kappa)$ for every $\kappa \in \R$, so has curvature $-\infty$; and at points from \emph{iii)} it is not $CAT(\kappa)$ for any $\kappa \in \R$; hence the curvature is $+\infty$. It thus suffices to show that the points in category \emph{iii)} are non-generic, which follows easily from the fact that these must necessarily sit in a lower-dimensional subspace of $\bar{X}$.

The proof carries over to the subspace $\bar{Z}$.
\end{proof}

\begin{defn}[Injectivity neighborhood]
We call a $CAT(0)$ neighborhood $\mathscr{U}$ of a point $\bar{x} \in \bar{X}$ an \emph{injectivity neighborhood} of $\bar{x}$. 
\end{defn}

Based on the above, we are now ready to prove:

\begin{proof}[Proof of theorem~\ref{qedcurv}]
\begin{itemize}
\item[i)] The QED case: Since $\bar{X}$ is locally $CAT(0)$ at generic points $\bar{x}$, the curvature of $\bar{X}$ is non-positive in a neighborhood $U$ of $\bar{x}$. At points $\bar{x} \in \bar{W}$, however, we will always find pairs of points $\bar{a}, \bar{b}$ arbitrarily close to $x$ with two geodesics joining them, just as in fig.~\ref{folding}. 
\item[ii)] The TED case: Consider a tree-shape $\tilde{T} \in \bar{X}$, represented by a point $x \in X$. Induce a second tree-shape $\tilde{T}'$ represented by $x + y_1 + y_2 \in X$, where $y_1, y_2 \in \prod_{e \in E} (\R^d)^n$ such that $y_1$ and $y_2$ have one non-zero coordinate, found in different edges, which are both nonzero edges in $x$. The topology of $\tilde{T}'$ is the same as of $\tilde{T}$. For any $n \in \N$, we can find $n$ TED geodesics $g_1, \ldots, g_n$ from $\tilde{T}$ to $\tilde{T}'$, where $g_i$ can be decomposed as $x \mapsto x + (i/n)y_1 \mapsto x + (i/n)y_1 + y_2 \mapsto x + y_1 + y_2$. Thus, there are infinitely many TED geodesics from $\tilde{T}$ to $\tilde{T}'$, and $(\bar{X}, \bar{d}_1)$ does not have locally unique geodesics anywhere. As a consequence, its curvature is unbounded everywhere \cite[proposition II 1.4]{bridsonhaef}.
\hfill \qed
\end{itemize}
\let\qed\relax
\end{proof}

The practical meaning of theorem~\ref{qedcurv} is that a) we can use techniques from metric geometry to search for QED averages, b) as we are about to see, for datasets contained in an injectivity neighborhood, there exist unique means, centroids and circumcenters for the QED metric, and c) the same techniques cannot be used to prove existence or uniqueness of prototypes for the TED metric, even if they were to exist. In fact, any geometric method which requires bounded curvature~\cite{karcher,bridsonhaef,phylogenetic} will fail for the TED metric. \emph{This motivates our study of the QED metric.}

\subsection{Curvature in the space of unordered tree-shapes}

It is easy to prove that the same results also hold for unordered tree-shapes:

\begin{thm} \label{unordered_curvature}
The space $(\bar{\bar{X}}, \bar{\bar{d}}_2)$ of unordered trees with the QED metric is generically non-positively curved. With the TED metric $\bar{\bar{d}}_1$, however, $\bar{\bar{X}}$ has everywhere unbounded curvature, geodesics are nowhere locally unique and neither are any of the types of average tree discussed in this paper. The same holds in the subspace $\bar{\bar{Z}} \subset \bar{\bar{X}}$, defined in theorem~\ref{unorderedthm}. \hfill \qed
\end{thm}

\subsection{Means and related statistics for tree-like shapes} \label{meansection}

In this section we use what we learned in the previous section to show that, given the QED metric on a space of tree-like shapes, we can find various forms of average shape in the space of ordered tree-like shapes, assuming that the data lie within an injectivity neighborhood.

There are many competing ways of defining central elements given a subset of a metric space. We discuss several: the \emph{circumcenter} considered in \cite{bridsonhaef}, the \emph{centroid} considered, among other places, in \cite{phylogenetic}, and the \emph{mean} \cite{karcher}.

The problem of existence and uniqueness of averages can be attacked using convex functions. Recall that a function $f \colon [a, b] \to \R$ is convex if $f((1-s)t + st') \le (1-s)f(t) + sf(t')$ for all $s \in [0,1]$ and $t, t' \in [a,b]$. If we can replace $\le$ with $<$ whenever $s \in ]0, 1[$, then $f$ is \emph{strictly} convex. Convex functions have minimizers, which are unique for strictly convex functions. Hence, existence and uniqueness of averages can be proven by expressing them as minimizers of strictly convex functions. 

We say that a function $f \colon X \to \R$ on a geodesic metric space $X$ is (strictly) convex if for any two points $x, y \in X$ and any geodesic $\gamma \colon [0, l] \to X$ from $x$ to $y$, the function $f \circ \gamma$ is (strictly) convex. We shall make use of the following standard properties of convex functions:

\begin{lem} \label{convexlemma}
\begin{itemize}
\item[i)] If $f \colon \R \to \R$ and $g \colon \R \to \R$ are both convex, $g$ is monotonous and increasing, and $g$ is strictly convex, then $g \circ f$ is strictly convex. 
\item[ii)] If $f \colon \R \to \R$ and $g \colon \R \to \R$ are both convex, then $g + f \colon \R \to \R$ is also convex. If either $f$ or $g$ is strictly convex, then $g + f$ is strictly convex as well. \hfill \qed
\end{itemize}
\end{lem}

The \emph{mean} of a finite subset $\{x_1, \ldots x_s\}$ in a metric space $(X, d)$ is defined as in eq.~\ref{meandef}, and is called the \emph{Fr\`{e}chet mean}. Local minimizers of eq.~\ref{meandef} are called \emph{Karcher means}.

The following result follows from a more general theorem by Sturm~\cite[Proposition 4.3]{sturm}; we include the basic version of the proof here for completeness.
\begin{thm} \label{existencemeans}
Means exist and are unique in $CAT(0)$-spaces.
\end{thm}

\begin{proof}
The function $d_y \colon Y \to \R$ given by $d_y(x) = d(x, y)$ is convex for any fixed $y \in Y$ by \cite[Proposition II.2.2]{bridsonhaef}, so the function $d_y^2$ is \emph{strictly} convex by lemma~\ref{convexlemma} i). But then $D = \sum_{i = 1}^s d_{x_i}^2$ is strictly convex by lemma~\ref{convexlemma} \emph{ii)}, and a mean is just a minimizer of the function $D$. The function $D$ is coercive, so the minimizer exists. Since $D$ is strictly convex, the minimizer is unique.
\end{proof}

We also consider two other types of statistical ''prototypes'' for a dataset, namely \emph{circumcenters} and \emph{centroids}. These are both well-known in the context of metric geometry and $CAT(\kappa)$ spaces.

\begin{defn}
\begin{itemize}
\item[a)] {\bf Circumcenters.} Consider a metric space $(Y, d)$ and a bounded subset $Z \subset Y$. There exists a unique smallest closed ball $\bar{B}(c_Z, r_Z)$ in $Y$ which contains $Z$; the center $c_Z$ of this ball is the \emph{circumcenter} of $Z$.
\item[b)] {\bf The centroid of a finite set.} Let $X$ be a uniquely geodesic metric space (a metric space where any two points are joined by a unique geodesic). The centroid of a set $S \subset X$ of $n$ elements is defined recursively as a function of the centroids of subsets with $n-1$ elements as follows: Denote the elements of $S$ by $s_1, \ldots, s_n$. If $S$ contains two elements, $|S| = 2$, the centroid $c(S)$ of $S$ is the midpoint of the geodesic joining $s_1$ and $s_2$. If $|S| = n$, define $c^1(S) = \{c(S') : |S'| = n-1\}$, which is a set with $n$ elements. Similarly, for larger $k$, $c^k(S) = c^1(c^{k-1}(S))$. All these sets have $n$ elements. If the elements of $c^k(S)$ converge to a point $c \in X$ as $k \to \infty$, then we say that $c = c(S)$ is the \emph{centroid} of $S$ in $X$.
\end{itemize}
\end{defn}

Based on the theory of $CAT(\kappa)$ spaces and our results for means, we have for the set of tree-like shapes:

\begin{thm} \label{all_averages}
Endow $\bar{X}$ with the QED metric $\bar{d}_2$. A generic point $\bar{x} \in \bar{X}$ has a neighborhood $U$ such that sets contained in $U$ have unique means, centroids and circumcenters.

\emph{The same statistical properties also hold for the QED metric on unordered tree-shapes:} Generic points in the space of unordered tree-shapes with the QED metric $(\bar{\bar{X}}, \bar{\bar{d}}_2)$ have neighborhoods within which means, circumcenters and centroids exist and are unique. 

For the TED metric, these are not unique.

The same results hold in the restricted tree-spaces $\bar{Z}$ and $\bar{\bar{Z}}$, defined in definition~\ref{subspacedef} and theorem~\ref{unorderedthm}.
\end{thm}

\begin{proof}
First consider $\bar{X}$ and $\bar{\bar{X}}$ with the QED metric. By theorems~\ref{qedcurv} and~\ref{unordered_curvature}, $\bar{X}$ and $\bar{\bar{X}}$ are both locally $CAT(0)$ spaces, and by theorems~\ref{orderedthm} and~\ref{unorderedthm} they are both complete metric spaces. 

We have seen in theorem~\ref{existencemeans} that means exist and are unique in $CAT(0)$ spaces, so the statement holds for means. By \cite[proposition~2.7]{bridsonhaef}, any subset $Y$ of a complete $CAT(0)$ space has a unique circumcenter. Hence, the statement holds for circumcenters. Similarly, by \cite[theorem~4.1]{phylogenetic}, finite subsets of $CAT(0)$ spaces $X$ have centroids (unique by definition), so the statement holds for centroids.

We turn to the TED metric. By definition, for any $2$-point dataset, all these notions of mean reduce to finding the midpoint of a geodesic connecting the two points. We know that geodesics and midpoints are not unique in the TED metric. This ends the proof.
\end{proof}

\subsection{The injectivity neighborhood} \label{injectivity}

We have shown that the local curvature is nonpositive almost everywhere in $\bar{X}$, which makes $\bar{X}$ well suited for geometric definitions of statistical properties. However, our notion of ''almost everywhere'' is strongly tied to (maximal) dimensionality, which again is strongly tied to the topological structure of the maximal tree $\mathscr{T}$. One consequence is that the injectivity neighborhoods in $\bar{X}$ are rather small, as we are about to see through examples. In this section we impose natural constraints on tree-space that allow us to increase the size of the injectivity neighborhoods, and make a conjecture for future expansion beyond the use of $CAT(0)$-spaces.

Consider the following two examples; we thank the anonymous reviewer for the first!

\begin{ex}
\begin{itemize}
\item[1)] Consider the tree-shapes $T_1$ and $T_2$ shown in fig.~\ref{ref_ex_1}, left. These two tree-shapes are joined by two geodesics, and thus $T_1$ and $T_2$ are not contained in the same injectivity neighborhood in $\bar{X}$. However, in a suitably chosen $\bar{Z}$, they can be.
\item[2)] Consider the tree-shape $\tilde{T}$ in the space of unordered tree-shapes with attributes in $\R^2$, spanned by the maximal tree $\mathscr{T}$ as in fig.~\ref{ref_ex_2}, right. Arbitrarily close to $\tilde{T}$ we will find two trees $T_1$ and $T_2$ which are joined by two geodesics, as shown in the figure.
\end{itemize}
\end{ex}

\begin{figure}
\centering
\includegraphics[width=0.7\linewidth]{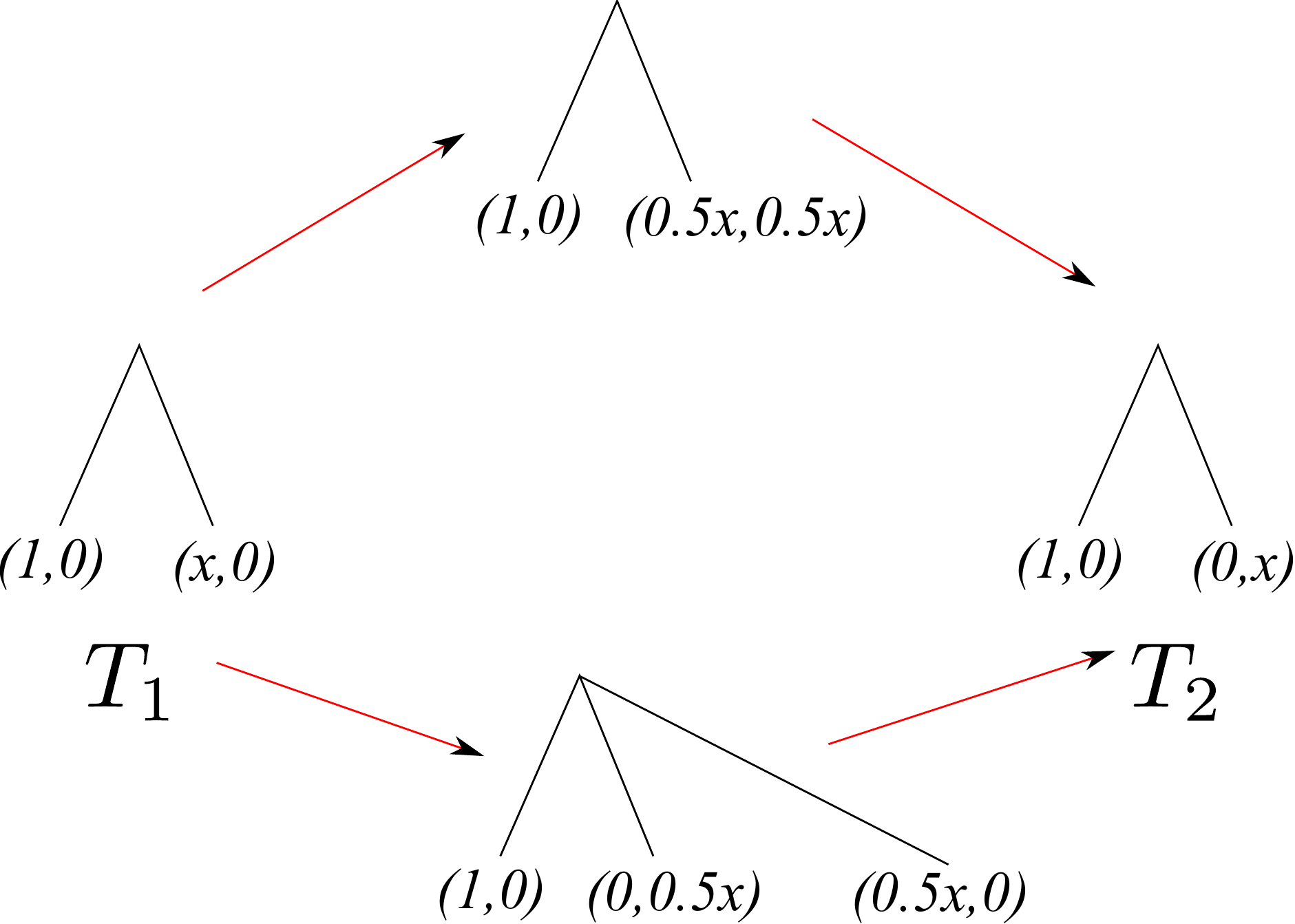}
\caption{Two trees joined by two different geodesics in $\bar{X}$.}
\label{ref_ex_1}
\end{figure}

\begin{figure}
\centering
\includegraphics[width=0.7\linewidth]{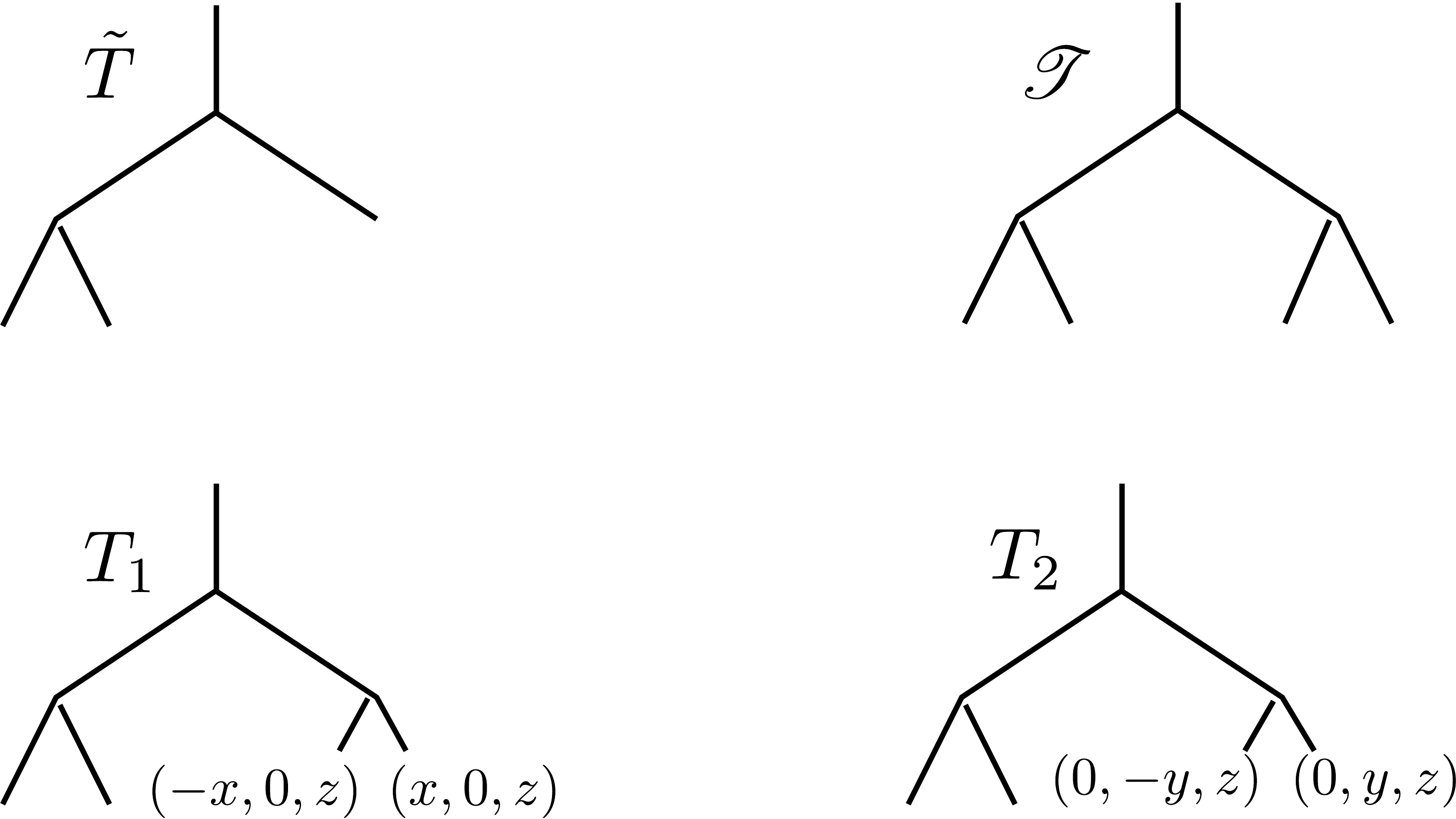}
\caption{Arbitrarily close to $\tilde{T}$ are two other trees $T_1$ and $T_2$ who cannot be joined by a unique geodesic.}
\label{ref_ex_2}
\end{figure}

As a illustrated by these examples, the $CAT(0)$ injectivity neighborhoods in the shape-space $\bar{X}$ can be very small, mainly containing trees whose topology is the same as $\mathscr{T}$. However, we can obtain much larger injectivity neighborhoods by restricting to the natural subspaces $\bar{Z} \subset \bar{X}$ and $\bar{\bar{Z}} \subset \bar{\bar{X}}$ as in definition~\ref{subspacedef}. As shown in theorem~\ref{all_averages}, $\bar{Z}$ and $\bar{\bar{Z}}$ have the same nice geometric properties as $\bar{X}$ and $\bar{\bar{X}}$, and $Z$ can be chosen so that the injectivity neighborhoods are bigger than in $\bar{X}$ and $\bar{\bar{X}}$ by avoiding situations as in the examples above. Radius is not a good measure for the size of an injectivity neighborhood, as a tree may contain both small branches, which do not have much room to vary, as well as large branches, which will be allowed to move more throughout such a neighborhood. However, any convex neighborhood which does not contain points of curvature $+ \infty$, will be an injectivity neighborhood.

All the results from section~\ref{meansection} hold at generic points, within an injectivity neighborhood where the $CAT(0)$ property holds. We have seen examples of points where local $CAT(0)$-property must fail, illustrating the limitaitons of $CAT(0)$ techniques for analyzing general trees. However, most of the situations where the $CAT(0)$ property fails are highly non-generic, and we conjecture that more general results can be proven:

\begin{con} \label{conjecture}
For a generic set of points $x_1, \ldots, x_n$ in $\bar{X}$, $\bar{\bar{X}}$, $\bar{Z}$ or $\bar{\bar{Z}}$, means exist and are unique.
\end{con}

\subsection{Comparison of QED and TED} \label{fundamental}

As we have seen in theorems~\ref{qedcurv}, the QED metric gives locally non-positive curvature at generic points, while the TED metric gives unbounded curvature everywhere on $\bar{X}, \bar{\bar{X}}, \bar{Z}$ and $\bar{\bar{Z}}$. Equivalently, geodesics are locally unique almost everywhere in the QED metric, while nowhere locally unique in the TED metric. As emphasized by theorem~\ref{all_averages}, this means that we cannot imitate the classical statistical procedures on shape spaces using the TED metric, while for the QED metric, we can.

Note, moreover, that the QED metric is the quotient metric induced from the Euclidean metric on the pre-shape space $X$, making it the natural choice of metric seen from the shape space point of view.


From a computational point of view, the TED metric has nice local-to-global properties. If the trees $T_1$ and $T_2$ are decomposed into subtrees $T_{1,1}, T_{1,2}$ and $T_{2,1}, T_{2,2}$ as in fig.~\ref{tedsplitting}, such that the geodesic from $T_1$ to $T_2$ restricts to geodesics between $T_{1,1}$ and $T_{2,1}$ as well as $T_{1,2}$ and $T_{2,2}$, then $d(T_1, T_2) = d(T_{1,1}, T_{2,1}) + d(T_{1,2}, T_{2,2})$. Many dynamic programming algorithms for TED use this property, and the same does not hold for the QED metric, due to the square root.

\begin{figure}
\centering
\subfigure[]{
\label{tedsplitting} 
\includegraphics[width=0.55\linewidth]{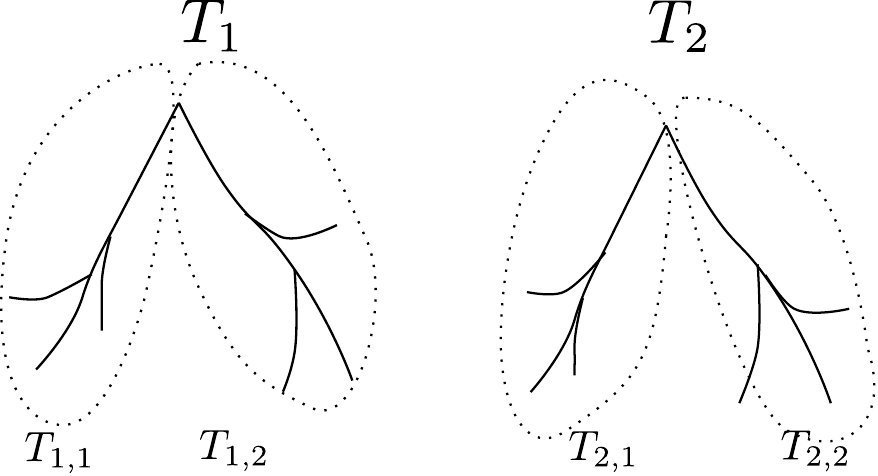}
}
\hspace{1mm}
\subfigure[]{
\label{tedqedcomparison}
\includegraphics[width=0.3\linewidth]{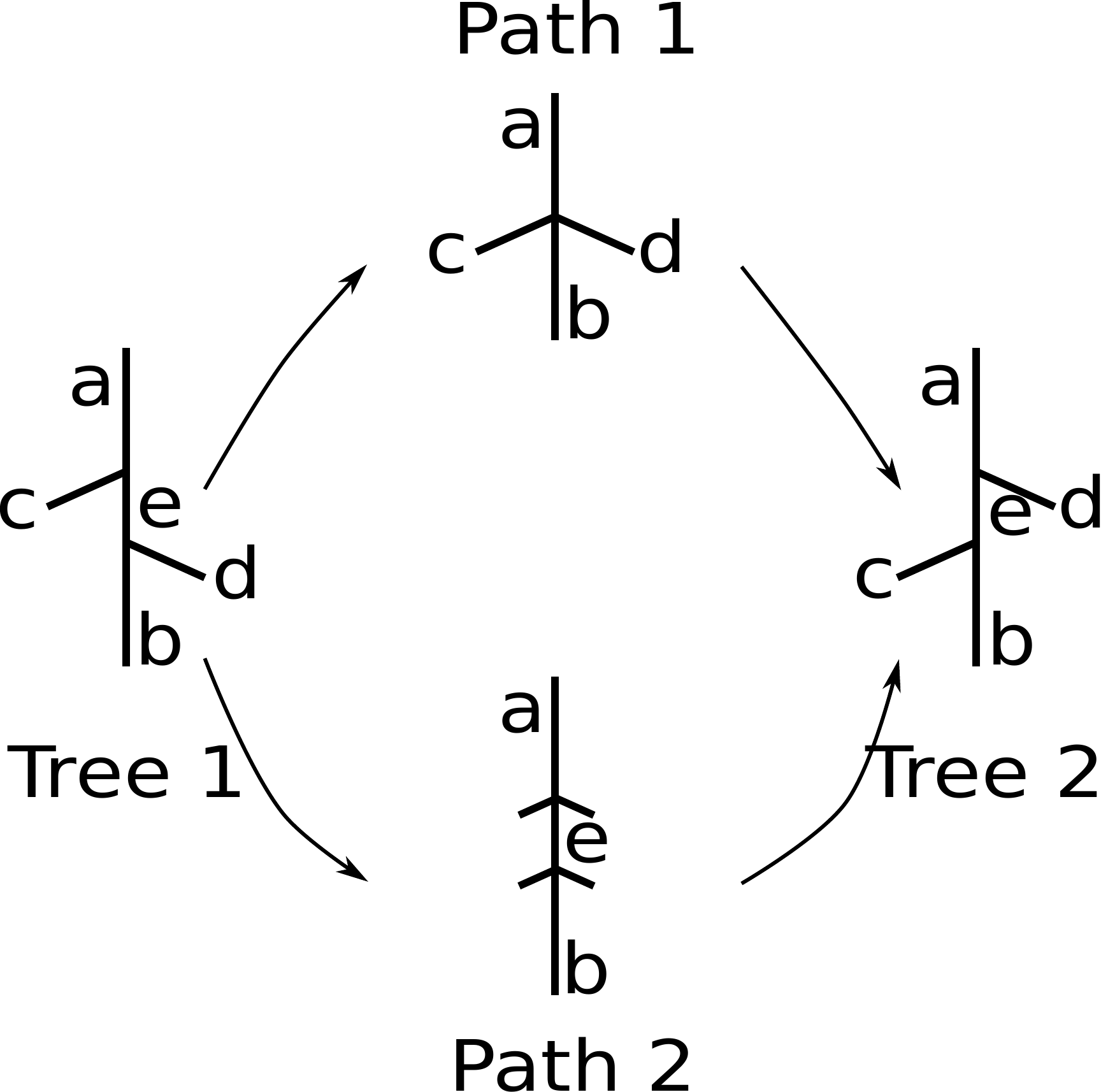}
}
\caption{(a) Local-to-global properties of TED: $\bar{d}_1(T_1, T_2) = \bar{d}_1(T_{1,1}, T_{2,1}) + \bar{d}_1(T_{1,2}, T_{2,2})$. (b) Two options for structural transition in a path from Tree 1 to Tree 2.}
\end{figure}


For a qualitative comparison of QED and TED, we compare the geodesics defined by the TED and QED metrics between small, simple trees. Consider the two tree-paths in fig.~\ref{tedqedcomparison}, where the edges are endowed with scalar attributes $a, b, c, d, e \in \R_+$ describing edge length. Computing the costs of the two different paths in both metrics, we find the shortest (geodesic) path. 

Path $1$ indicates a matching left and right side edges $c$ and $d$, while Path $2$ does not make the match. The cost of Path $1$ is $2e$ in both metrics, while the cost of Path $2$ is $2\sqrt{c^2 + d^2}$ in the QED metric and $2(c+d)$ in the TED metric. Thus, TED will identify the $c$ and $d$ edges whenever $e^2 \le c^2 + 2cd + d^2$, while QED makes the identification whenever $e^2 \le \textrm{\textonehalf} (c^2 + d^2)$. Thus, TED will be more prone to internal structural change than QED. 

The same occurs in the comparison of TED and QED matching in figs.~\ref{qedmatching} and~\ref{tedmatching}. Although the TED is more prone to matching trees with different tree-topological structures, the edge matching results are similar, as is expected, since the metrics are closely related.

\begin{figure}
\centering
\subfigure[Matching in the QED metric.]{\label{qedmatching} \includegraphics[width=0.58\linewidth]{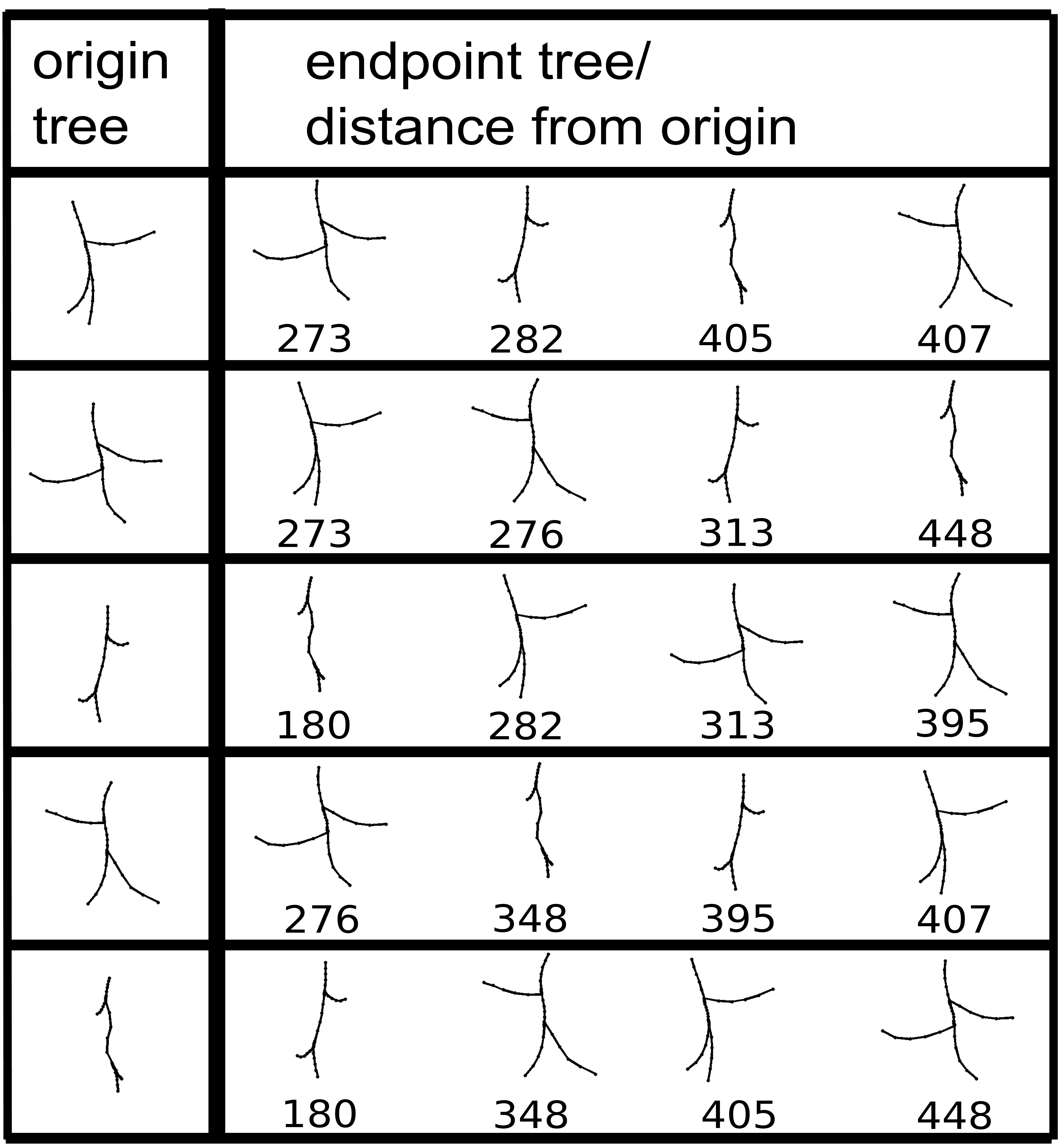}} 
\hspace{2mm}
\subfigure[Matching in the TED metric.]{\label{tedmatching} \includegraphics[width=0.58\linewidth]{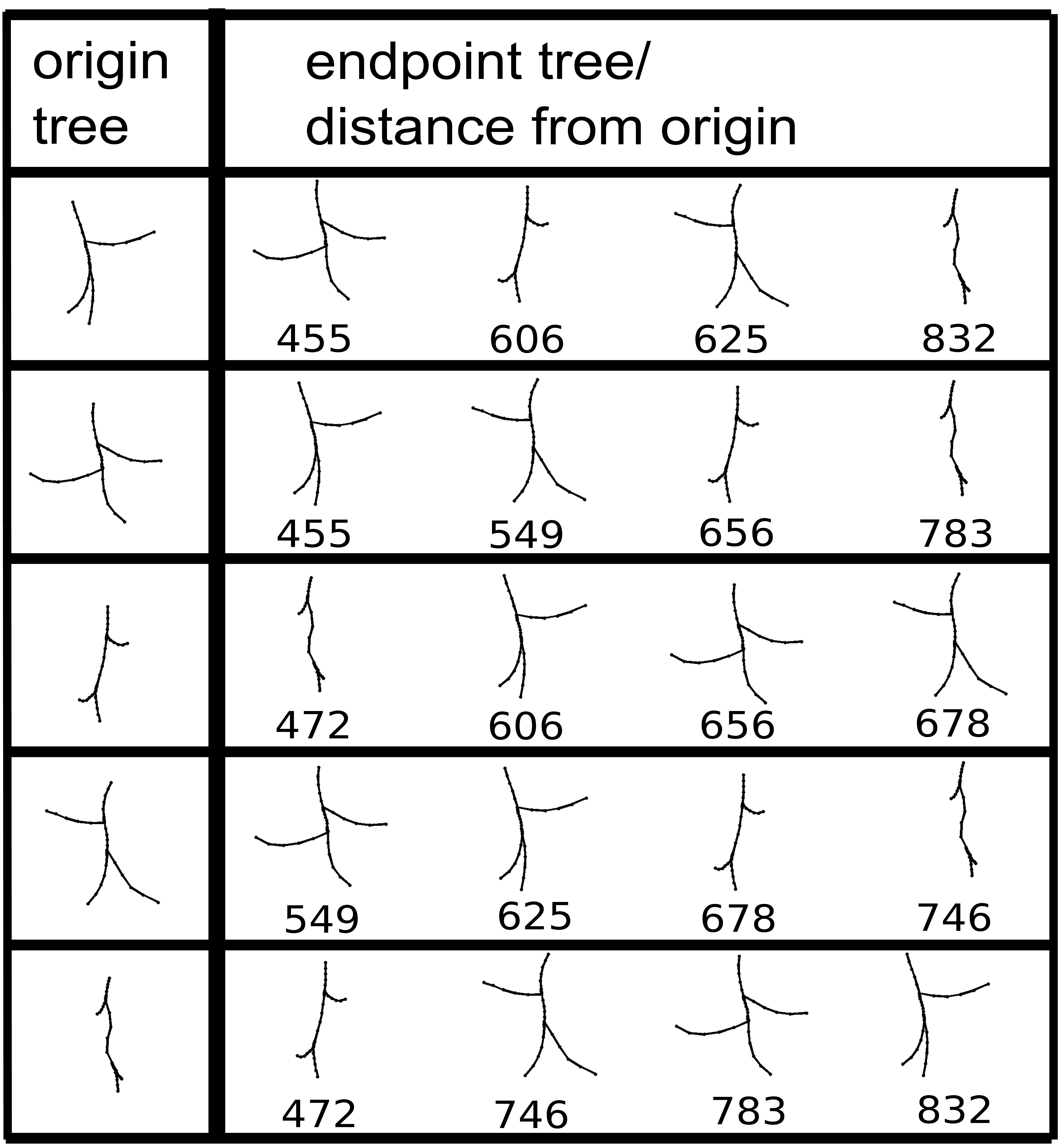}}
\caption{Given a set of five data trees, we match each to the four others in both metrics.}
\end{figure}

\section{Computation and complexity} \label{complexitysection}

In general, computational complexity is a problem for both TED and QED. Computing TED distances between unordered trees is NP-complete \cite{zhangetal}, and we conjecture that the QED metric is NP-complete as well. We can, however, often use geometry and prior knowledge (e.g., anatomy) to find efficient approximations. \emph{Trees appearing in applications are often not completely unordered, but are semi-labeled}.

\begin{ex}[Semi-labeling of the upper airway tree] \label{semilabel}
Most airway trees have similar, but not necessarily identical, topological structure, where several branches have names and can be identified by experts. 

The top generations of the airway tree serve very clear purposes in terms of anatomy. The root edge is the trachea; the second generation edges are the left and right main bronchi; the third generation edges lead to the lung lobes. As these are easily identified, we find a semi-labeling of the airway tree, and use it to simplify computations of inter-airway geodesics in section~\ref{pulmonarysection}.
\end{ex}

%

\subsection{QED approximation and implementation}

Since the decomposition strategy used for dynamic programming in TED is not available for QED, an approximation of QED is used in our experiments. Many anatomical trees have a somewhat fixed overall structure. For these, it is safe to assume that the number of internal structural transitions found in a geodesic deformation is low, and that the geodesics pass through identified subspaces of low codimension. The latter assumption is equivalent to assuming that the trees appearing throughout the geodesic deformation have nodes of low degree. For instance, for the airway trees studied in section~\ref{pulmonarysection} below, we find empirically that it is enough to allow for one structural change in each lobar subtree, which has at most a trifurcation. Recall the definition of the metric from eq.~\ref{metricdefn}; the approximation consists of imposing upper bounds $K$ on the number $k$ and $D$ on the degrees of internal vertices appearing in eq.~\ref{metricdefn}, respectively, giving:
\begin{equation} \label{metricdefn2}
\bar{d}_{approx, K}(\bar{x}, \bar{y}) = \inf_{k \le K} \left\{ \sum_{i=1}^k
    d(x_i, y_i) | x_1 \in \bar{x}, y_i \sim x_{i+1}, y_k \in \bar{y}\right\},
\end{equation}
where all $x_i$ and $y_i$ have vertex degrees at most $D$. Geometrically, $k$ is the number of Euclidean segments concatenated to form the geodesic. Bounding $k$ is equivalent to bounding the number of internal topological transitions throughout the geodesic. In fact, $k = 1 \  + $ \emph{the number of internal topological transitions in the geodesic}.


All edges are translated to start at $0 \in \R^d$, and represented by a fixed number of landmark points (in our case $6$) evenly distributed along the edge, the first at the origin. The distance between two edge attributes $v_1, v_2 \in (\R^d)^5$ is the Euclidean distance $\|v_1 - v_2\|$.


We approximate of the QED metric using Algorithm~\ref{planaralg}; see Appendix B in the supplemental material for details. The complexity of Algorithm 1 is 
\[
O(n^{(D - 2)(K-1)})\cdot O(optim(n,K)),
\]
where $n$ is the number of internal vertices, $K$ is the bound on $k$, $D$ is the bound on vertex degree and $optim(n,K)$ is the optimization in line $8$. For the unordered airway trees in section~\ref{pulmonarysection}, we combine Algorithm 1 with a complete search through the set of branch orderings. Computing QED distances through a complete search is not optimal, and finding more exact and efficient algorithms is a nontrivial research problem.

%

\begin{algorithm}
\caption{Computing approximate QED distances between ordered, rooted trees with up to $k = K-1$ structural transitions through trees of order at most $D$}
\begin{algorithmic}[1]
\STATE $x$, $y$ planar rooted depth $n$ binary trees
\STATE $\mathbf{S} = \{S\}$ set of ordered identified pairs $S = \{S_1, S_2\}$ of subspaces of $X$ corresponding to internal topological changes through trees of order at most $D$, corresponding to a subspace $S$ of $\bar{X}$, s.t.~if $\{S_1, S_2\} \in \mathbf{S}$, then also $\{S_2, S_1\} \in \mathbf{S}$. 
\FOR{$\tilde{S} = \{S^1, \ldots, S^s\} \subset \mathbf{S}$ with $|\tilde{S}| \le k$}
\FOR{$p^i \in S^i$ with representatives $p^i_1 \in S^i_1$ and $p^i_2 \in S^i_2$} 
\STATE $p = (p^1, p^2, \ldots, p^s)$
\STATE $f(p) = \min\{ d_2(x, p_1^1) + \sum_{j=1}^{s-1} d_2(p^j_2, p^{j+1}_1) + d_2(p^s_2, y) \}$
\ENDFOR
\STATE $d_{\tilde{S}} = \min \left\{
\begin{array}{c|c}
f(p) & 
\begin{array}{c} 
p = (p^1, \ldots, p^s), p^i \in S^i,\\
 \tilde{S} = \{S^1, \ldots, S^s\}
\end{array}
\end{array}
 \right\}$
\STATE $p_{\tilde{S}} = \{p^i_1, p^i_2\}_{i=1}^s = \textrm{argmin} f(p)$
\ENDFOR
\STATE $d = \min \{d_{\tilde{S}} | \tilde{S} \subset \mathbf{S}, |\tilde{S}| \le k\}$
\STATE $p = \{p_1, p_2\}_{i=1}^s = \{p_{\tilde{S}} | d_{\tilde{S}} = d\}$
\STATE geodesic $ = g  = \{x \rightarrow p^1_1 \sim p^1_2 \rightarrow p^2_1 \sim p^2_2 \rightarrow \ldots \rightarrow p^s_1 \sim p^s_2 \rightarrow y \}$
\RETURN{$d, g$}
\end{algorithmic}
\label{planaralg}
\end{algorithm}

\section{Experimental results} \label{experimentalsection}

The QED metric is new, whereas the properties of the TED metric are well known \cite{editshock}. Our experimental results on real and synthetic data illustrate the geometric properties of the QED metric. The experiments on airway trees in section~\ref{pulmonarysection} show, in particular, that it is feasible to compute geodesics between real, $3D$ data trees.

\subsection{Synthetic planar trees of depth 3}

Movies illustrating geodesics between planar depth $3$ trees, as well as a table illustrating a tree-shape matching experiment for a set of synthetic planar depth $3$ trees, are found on the web page \url{http://image.diku.dk/aasa/tree_shape/planar.html}. The movies and matchings show the geometrically intuitive behavior of the geodesic deformations. We see that the intermediate structures resemble the geodesic endpoint trees in a reasonable manner, and show the ability of the geodesics to handle internal topological differences. These are among the wanted properties from the model listed in the beginning of section~\ref{geometrysection}.

Moreover, computed mean and centroid trees are shown in fig.~\ref{dataset_means}; these are presented in greater detail in~\cite{feragen_iccv11}. The QED average trees clearly represent the main common properties of the dataset trees.

\subsection{Results in 3D: Pulmonary airway trees} \label{pulmonarysection}

As proof-of-concept experiments on real data we compute means for sets of pulmonary airway trees from the EXACT'09 airway segmentation competition~\cite{exact}. The airway trees were first segmented from low dose screening computed tomography (CT) scans of the chest using a voxel classification based airway tree segmentation algorithm~\cite{lo}. The centerlines were extracted from the segmented airway trees using a modified fast marching algorithm based on~\cite{schlatholter2002}. The method gives a tree structure directly through connectivity of parent and children branches. For simplicity, we only consider the upper airway tree down to the lobe branches.

Means are computed using both the QED metric (details found in \cite{feragen_iccv11}) and using the TED method proposed by Trinh and Kimia~\cite{trinh}. The dataset is shown in fig.~\ref{airway_dataset}, and the results are seen in fig.~\ref{airway_means}. We clearly see how the choice of TED geodesic makes the TED mean vulnerable to noise in the form of missing branches among the dataset trees, giving a mean shape whose structure is too simple. The QED mean, on the other hand, represents the data well.

\begin{figure}
\centering
\includegraphics[width=0.7\linewidth]{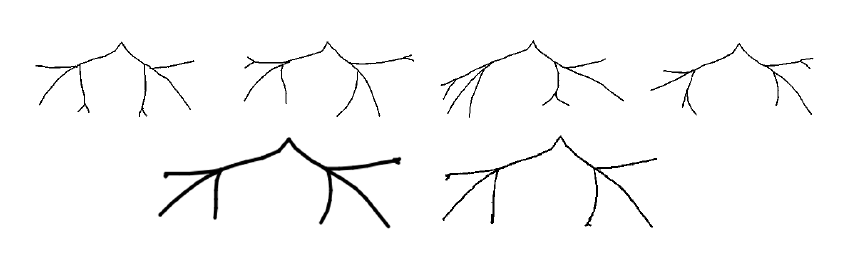}
\caption{Top: A small set of synthetic planar trees. Bottom left: The mean tree. Bottom right: The centroid tree.}
\label{dataset_means}
\end{figure}

\begin{figure}[!t]
\centering
\includegraphics[width=0.8\linewidth]{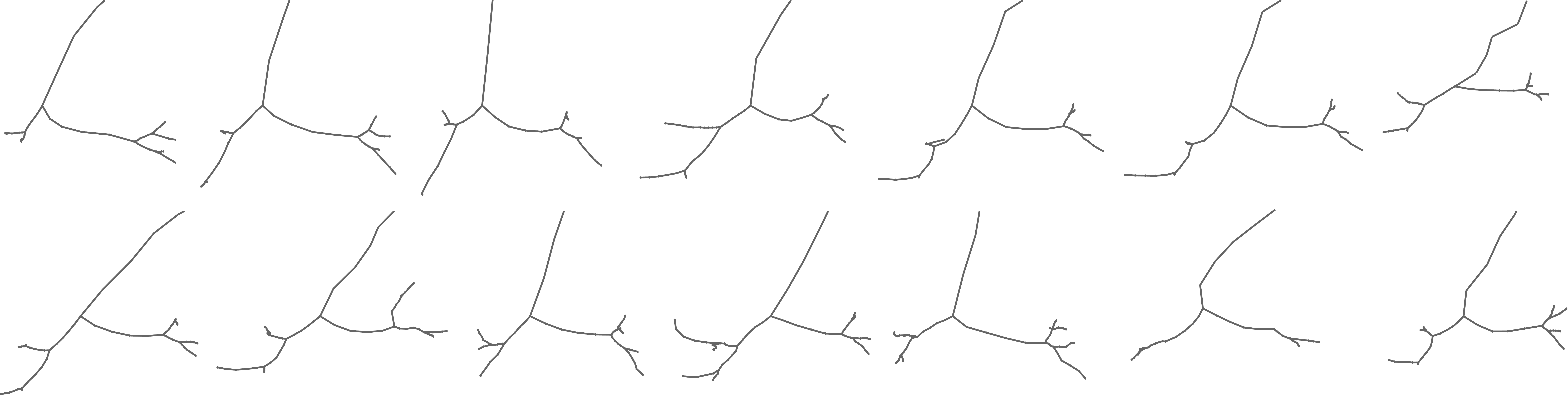}
\caption{Subtrees of $14$ airway trees from the EXACT'09 test set. \cite{exact}}
\label{airway_dataset}
\end{figure}

\begin{figure}
\centering
\includegraphics[width=0.35\linewidth]{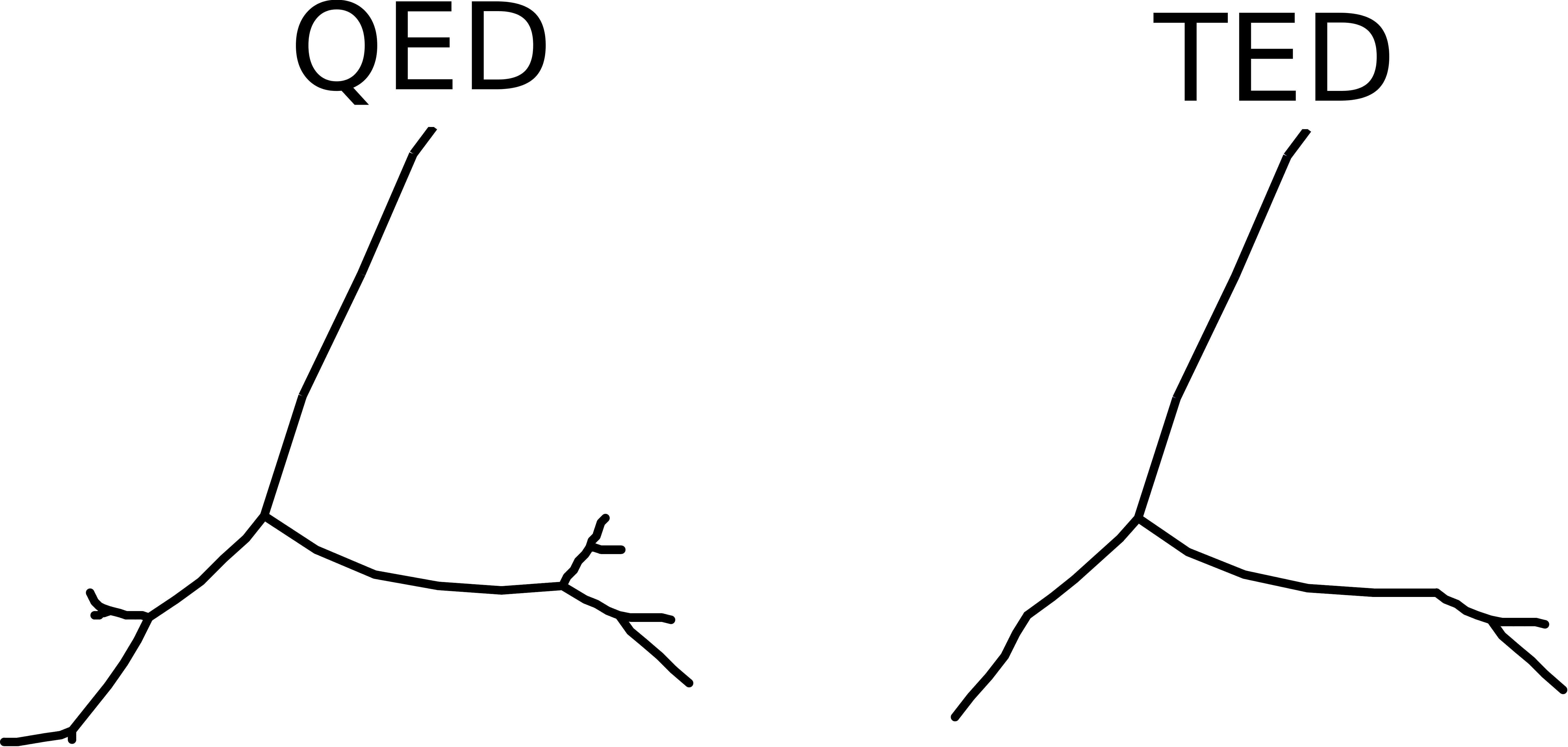}
\caption{The QED and TED mean airway trees. The TED mean misses the upper lobe branches on both sides, caused by the fact that the $13$th data tree misses these.}
\label{airway_means}
\end{figure}

\section{Discussion and conclusion}

Starting from a purely geometric point of view, we define a shape space for tree-like shapes. This intuitive geometric framework allows us to simultaneously take both global tree-topological structure and local edgewise geometry into account. We study two metrics on this shape space, TED and QED, which both give the shape space a geodesic structure (theorems~\ref{orderedthm} and~\ref{unorderedthm}). The framework is developed for tree-shapes, but carries over to other geometric trees with vector attributes.

QED is the geometrically most natural metric, and turns out to have properties which are essential for statistical shape analysis. Through a geometric analysis of the tree-shape framework, we show that the QED metric gives local uniqueness of geodesics and local existence and uniqueness for three versions of average shape, namely the mean, the circumcentre and the centroid (theorem~\ref{all_averages}). Our analysis gives new insight into the metric space defined by TED, and explains why the problem of finding TED-averages is ill-posed.

Both metrics are generally NP complete to compute for $3D$ trees. We explain how semi-labeling schemes can be used to handle complexity problems, and illustrate this by computing QED means for sets of trees extracted from pulmonary airway trees as well as synthetic planar data trees. The QED lacks an obvious decomposition strategy for dynamic programming, and involves an optimization term which makes it more time consuming to compute than the TED metric. The two metrics have similar matching properties, and thus, for applications which do not require a unique geodesic (e.g., clustering) the TED metric is just as suitable as the QED metric. However, for statistical computations, the QED metric is the only suitable metric. This is emphasized by our experiments, where computed QED and TED means illustrate that the TED means are missing important structure.

Future research will be centered around two points: Development of nonlinear statistical methods for the singular tree-shape spaces, and finding more efficient approximations and heuristics for the QED metric using both the tree geometry and the tree-space geometry. The latter will pave the road for computing averages and modes of variation for large, real $3D$ data trees. This is by no means trivial, due to the non-smooth structure of the tree-shape space and the complexity of computing exact distances and geodesics.
%
%

%
%


\begin{appendices}

\section{Measures induced on quotient spaces}

The purpose of this appendix is to explain the connection between genericity, defined as a topological property, and probability. In particular, we explain why a generic property can be assumed to hold with probability one, or almost everywhere, with respect to natural probability measures. The main point of the appendix is the contents of Theorem~\ref{mainthm}.

\begin{defn}[Generic property]
A property in a topological space $X$ is said to be \emph{generic} if it holds on a dense, open subset of $X$.
\end{defn}

\begin{defn}[$\sigma$-algebra]
  A collection $A = \{U\}$  of subsets of some set $X$ is a \emph{$\sigma$-algebra} if $A$
  has the following properties:
\begin{itemize}
\item $X \in A$.
\item If $U \in A$, then its complement is also an element: $U^C \in A$.
\item If $U_i \in A$ for a family $i = 1 \ldots \infty$, then the union is also an
  element: $\bigcup_{i=1}^\infty U_i \in A$.
\end{itemize}
\end{defn}

\begin{defn}[Measure]
  A \emph{measure} on a $\sigma$-algebra $A$ is a function $\mu \colon A \to \R_{\ge 0}$
  satisfying the $\sigma$-additivity property: if $\{U_i\}_{i = 1 .. \infty}$ is a
  countable family of pairwise disjoint subsets in $A$, then $\mu \left(\bigcup_{i =
      1}^\infty U_i \right) = \sum_{i = 1}^\infty \mu(U_i)$.
\end{defn}

\begin{defn}[Almost everywhere]
A property is said to hold \emph{almost everywhere} with respect to a probability measure $\mu$ on a space $X$ if it holds on a subset $U \subset X$ such that $\mu(X \setminus U) = 0$. This corresponds to holding with probability one, or \emph{almost surely}.
\end{defn}

Assume given a metric space $(X, d)$, and $\sigma$-algebra $A$ which contains the topology
induced by $d$, and measure $\mu$ on the sigma-algebra $A$. Assume also given an
equivalence class $\sim$ on $X$, and denote by $\bar{X} = X/\sim$ the quotient of $X$ with
respect to the equivalence. Denote by $\pi \colon X \to \bar{X}$ the projection onto the
quotient.

\begin{lem}
  The set $\bar{A} = \{\bar{U} \subset \bar{X} | \pi^{-1}(\bar{U}) \in A\}$ is a
  $\sigma$-algebra on $\bar{X}$.
\end{lem}

\begin{proof}
\begin{itemize}
\item $\pi^{-1}\bar{X} = X \in A$ so $\bar{X} \in \bar{A}$.
\item If $\bar{U} \in \bar{A}$ then $\pi^{-1}\bar{U} \in A$ so $\pi^{-1}(\bar{U}^C) =
  (\pi^{-1}\bar{U})^C \in A$, so $\bar{U}^C \in \bar{A}$.
\item If $\bar{U}_i \in \bar{A}$ for $i = 1 \ldots \infty$ then $\pi^{-1} \bar{U}_i \in A$
for all $i$, so $\pi^{-1} \left( \bigcup_{i = 1}^\infty U_i \right) = \bigcup_{i=1}^\infty
\pi^{-1} \bar{U}_i \in A$ so $\bigcup_{i=1}^\infty \bar{U}_i \in \bar{A}$.
\end{itemize}
\end{proof}

\begin{lem}
  The function $\bar{\mu} \colon \bar{A} \to \R_{\ge 0}$ given by $\bar{\mu}(\bar{U}) =
  \mu(\pi^{-1}(\bar{U}))$ is a measure on $\bar{A}$.
\end{lem}

\begin{proof}
  If $\{\bar{U}_i\}_{i=1}^\infty$ is a countable family of pairwise disjoint subsets in
  $\bar{A}$, then $\{\pi^{-1}\bar{U}_i\}_{i=1}^\infty$ is a countable family of pairwise
  disjoint subsets in $A$, and $\bar{\mu} \left(\bigcup_{i=1}^\infty \bar{U}_i\right) =
  \mu \left( \pi^{-1} \left( \bigcup_{i=1}^\infty \bar{U}_i \right) \right) = \mu \left(
    \bigcup_{i=1}^\infty \left( \pi^{-1} \bar{U}_i \right) \right) = \sum_{i=1}^\infty \mu
  \left(\pi^{-1} \bar{U}_i \right) = \sum_{i=1}^\infty \bar{\mu} \left(\bar{U}_i \right)$.
\end{proof}

\begin{prop}
  There are Lebesgue induced measures $\bar{m}_X$ and $\bar{\bar{m}}_X$ on the tree-spaces $\bar{X}$ and
  $\bar{\bar{X}}$, coming from the Lebesgue measure $m_X$ on $X$. For these measures, subsets of positive codimension have measure zero. In particular, generic properties hold almost everywhere with respect to the Lebesgue induced measure. \hfill $\square$
\end{prop}

Recall that we also work on a restricted subspace $Z \subset X$ which consists of trees of certain fixed topologies that are collapsed versions of the maximal tree $\mathscr{T}$. The subspace $Z$ is a finite union
\[
 \bigcup_{i=1}^n Z_i
\]
where each $Z_i$ is a $d$-dimensional linear subspace of $X$ corresponding to a certain fixed collapsed version of $\mathscr{T}$ described by its nonzero edge set $E_i \subset E$. We assume, moreover, that $Z_i \nsubseteq Z_j$ whenever $i \neq j$. Each $Z_i$ now has a Lebesgue measure $m_i$ of its own, and we obtain a Lebesgue induced measure $m_Z$ on $Z$ defined as 
\[
 m_Z(U) = \sum_{i=1}^n m_i(U \cap Z_i).
\]
The measure $m_Z$ gives rise to Lebesgue induced measures on $\bar{Z}$ and $\bar{\bar{Z}}$ as explained above.

\begin{prop}
  There are Lebesgue induced measures $\bar{m}_Z$ and $\bar{\bar{m}}_Z$ on the tree-spaces $\bar{Z}$ and $\bar{\bar{Z}}$, coming from the Lebesgue induced measure $m_Z$ on $Z$. If all $Z_i$ have equal dimension, subsets of positive codimension have measure zero. In particular, generic properties hold almost everywhere with respect to the Lebesgue induced measure. \hfill $\square$
\end{prop}

\begin{defn} \cite[p. 120]{rudin} Let $\mu$ and $\lambda$ be two measures on a
  $\sigma$-algebra $A$. The measure $\lambda$ is \emph{absolutely continuous} with respect to
  $\mu$ if and only if $\lambda(U) = 0$ for each $U \in U$ with $\mu(U) = 0$.
\end{defn}

\begin{thm} \label{mainthm}
  Let $\lambda$ be a probability measure which is absolutely continuous with respect to
  the Lebesgue induced measure $\bar{m}$ on one of the tree-space $\bar{X}$ or $\bar{\bar{X}}$, or of $\bar{Z}$ or $\bar{\bar{Z}}$ under the assumption that all $Z_i$ have equal dimension. Now any subset of positive
  codimension has zero measure with respect to $\lambda$.  In particular, generic properties hold almost surely, or with probability $1$, with respect to $\lambda$.  \hfill $\square$
\end{thm}

\begin{rem}
Note that the converse is not necessarily true: Properties which hold almost surely are not necessarily generic.
\end{rem}

\section{Implementation of approximate QED geodesic computation}

In this appendix, we extend the explanation of how the approximate QED distance and geodesic computation specified in section 4.2 in the main paper can be made in practice. We shall consider computation of distances in the space of ordered trees $\bar{X}$, as the complete search reduction from ordered trees to the unordered tree-space $\bar{\bar{X}}$ is not hard to implement.

Recall that the tree-space $\bar{X}$ is constructed from the Euclidean space 
\[
X = \prod_{e \in E} (\R^d)^n,
\]
consisting of all trees spanned by a certain maximal combinatorial tree $\mathscr{T} = (V, E, r)$. The construction is made by identifying different representations of the same tree, with the help of an equivalence relation $\sim$. Denote by $\bar{x}$ the equivalence class of a point $x \in X$. As discussed in the article, there is a standard way of defining a distance function on $\bar{X} = X/\sim$, given by the quotient pseudometric
\begin{equation} 
\bar{d}(\bar{x}, \bar{y}) = \inf \left\{ \sum_{i=1}^k d(x_i, y_i) | x_1 \in \bar{x}, y_i \sim x_{i+1}, y_k \in \bar{y}\right\},
\end{equation}
which is shown to be an actual metric.

In the paper, we suggest approximating the geodesic by bounding the number of visited equivalence classes:

\begin{equation}
\begin{array}{c}
\bar{d}_{approx, K}(\bar{x}, \bar{y}) =\\
\inf_{k \le K} \left\{ \sum_{i=1}^k d(x_i, y_i) | x_1 \in \bar{x}, y_i \sim x_{i+1}, y_k \in \bar{y}\right\},
\end{array}
\end{equation}
and in practice, in our experiments, we have used an approximation with $K$ either $2$ or $3$.

\subsection{Reformulation of tree-space identifications in terms of tree isomorphism} \label{reformulation}

The quotient space construction is extremely useful for understanding the geometry of tree-space and proving statistical properties of tree-space based on geometry. However, it is less intuitive how one should proceed in order to actually compute distances in tree-space. In this section, we shall explain how the tree-space gluings are related to the trees themselves.

Recall that two points in $X$ are equivalent under $\sim$ and glued together in the formation of $\bar{X}$ if they represent \emph{the same} ordered tree-shape. How can we determine whether two points in tree-space represent the same ordered tree-shape? This reduces to a tree isomorphism problem. More precisely:

\begin{itemize}
\item[i)] Assume given two points $x_1, x_2 \in X$.
\item[ii)] These points correspond to two combinatorial trees $t_1, t_2$, equal to or derived from $\mathscr{T}$, where some of the edges in $t_i$ may be endowed with zero attributes in $x_i$.
\item[iii)] Create two new combinatorial trees $t'_1, t'_2$ by deleting the edges in $t_1$ and $t_2$ with zero attributes in $x_1$ and $x_2$, and using the induced orders on siblings from $t_1$, $t_2$ to induce total orders on siblings in $t'_1$, $t'_2$.
\item[iv)] The question of whether $x_1$ and $x_2$ end up in the same identified subspace is now equivalent to the question of whether $t'_1$ and $t'_2$ are isomorphic as ordered, combinatorial trees. The question of whether $x_1$ and $x_2$ represent the same tree-shape, is equivalent to the question of whether a) $t'_1$ and $t'_2$ are topologically isomorphic, and b) if yes, whether their ordered sets of attributes are identical.
\item[v)] The observations above give us a way to judge whether two points in $X$ are equivalent under the equivalence relation $\sim$ on $X$; equivalently, whether they are identified in $\bar{X}$. This realization can be used to reformulate Algorithm 1 from the main paper in terms of the trees rather than in terms of the tree-space.
\item[vi)] A geodesic in $\bar{X}$ is composed of a set of linear stretches in $X$ between equivalence classes that belong to lower-dimensional identified subspaces. Each such stretch can change one endpoint tree to the other endpoint tree by i) gradually shrinking and removing edges, ii) adding and gradually extending edges, or iii) gradually changing edges which are ''present'' in both trees. All these changes are made simultaneously at constant speed throughout the linear stretch.
\item[vii)] Given two points $\bar{x}_1$ and $\bar{x}_2$ in $\bar{X}$, how can we find the shortest path joining them which consists of two Euclidean stretches? The length of the shortest such path is $\bar{d}_{approx, 2}(\bar{x}_1, \bar{x}_2)$. We shall look at this problem in section~\ref{twostretches} below.
\end{itemize}

\subsection{Paths with two Euclidean stretches} \label{twostretches}

In this section we study the problem of finding  $\bar{d}_{approx, 2}$.

\begin{figure*}
\centering
\includegraphics[width=0.9\linewidth]{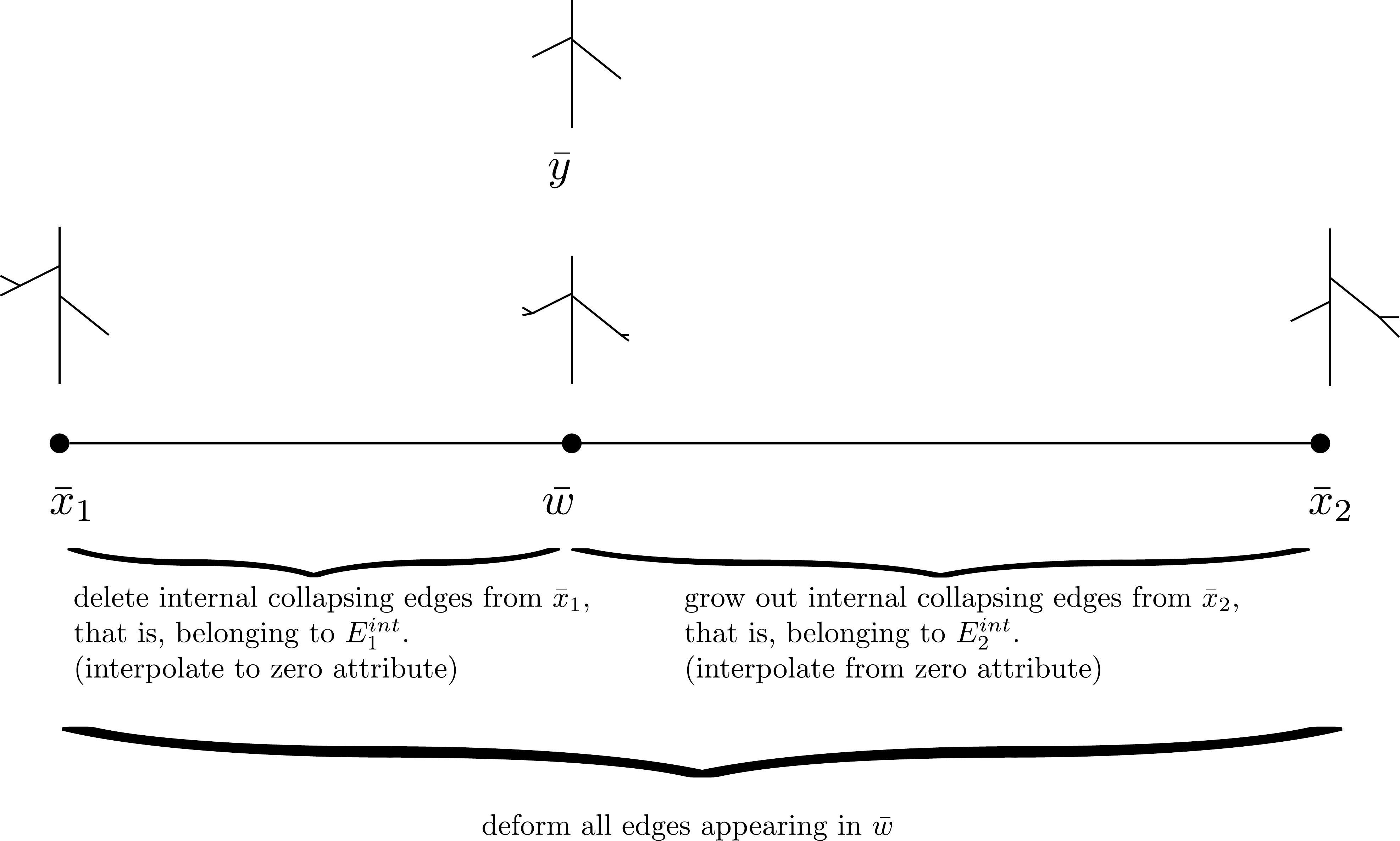}
\caption{Path consisting of two Euclidean stretches.}
\label{algorithm}
\end{figure*}

\begin{itemize}
\item Let $\bar{x}_1$, $\bar{x}_2$ be tree-shapes in $\bar{X}$. Let $\bar{y}$ be some tree-shape whose topological structure can be obtained as a collapsed version of both $\bar{x}_1$ and $\bar{x}_2$. 
\item The topology of the tree $\bar{y}$ is obtained from $\bar{x}_1$ (and similarly from $\bar{x}_2$) by collapsing two types of edges:
\begin{itemize}
\item internal edges (that is, there is some edge in the tree which is further from the root in $\bar{x}_1$ (or $\bar{x}_2$), which is not being collapsed), or
\item external edges (that is, all the edges in $\bar{x}_1$ (or $\bar{x}_2$) which are farther from the root are also collapsed).
\end{itemize}
\item Denote the sets of such edges as $E_1^{int}$, $E_1^{ext}$, $E_2^{int}$, $E_2^{ext}$, respectively.
\item We are looking for a point $\bar{w} \in \bar{X}$ which contains some such $\bar{y}$ as a subtree, such that the shortest two-stretch path from $\bar{x}_1$ to $\bar{x}_2$ passes through $\bar{w}$, and consists of the following types of deformations:
\begin{itemize}
\item deformation of all edges that collapse to an edge in $\bar{y}$; throughout the whole path.
\item gradually deforming all external collapsing edges in $\bar{x}_1$ so that when we reach $x_2$, their attributes have become zero; throughout the whole path.
\item gradually growing all external collapsing edges in $\bar{x}_2$ from zero attributes when the deformation starts in $x_1$; throughout the whole path
\item gradually deforming all internal collapsing edges in $\bar{x}_1$ so that when we reach $\bar{w}$, their attributes are zero; throughout the first stretch.
\item from $\bar{w}$, gradually grow all internal collapsing edges in $\bar{x}_2$; throughout the second stretch.
\item Thus, $\bar{w}$ contains the edges from $\bar{z}$ as well as the edges from $E_1^{ext}$ and $E_2^{ext}$.
\end{itemize}
Now the distance from $\bar{x}_1$ to $\bar{x}_2$ via $\bar{w}$ is:
\[
\begin{array}{c}
dist(\bar{x}_1, \bar{x}_2, \bar{w}) =\\

\underbrace{\sqrt{\sum_{e \in \bar{w}} \|(\bar{x}_1)_e - \bar{z}_e\|^2 + \sum_{e \in E_1^{int}} \|(\bar{x}_2)_e\|^2}}_{\textrm{length of the first stretch from } \bar{x}_1 \textrm{ to } \bar{w}} + \\
\underbrace{\sqrt{\sum_{e \in \bar{w}} \|(\bar{x}_2)_e - \bar{z}_e\|^2 + \sum_{e \in E_2^{int}} \|(\bar{x}_2)_e\|^2}}_{\textrm{length of the first stretch from } \bar{w} \textrm{ to } \bar{x}_2},
\end{array}
\]
where $\|(\bar{x}_i)_e - \bar{z}_e\| = \|\bar{z}_e\|$ in case $e$ is not an edge in $\bar{x}_i$, $i=1, 2$.
\item This is illustrated in figure~\ref{algorithm}.
\item Now we can define $\bar{d}_{approx, 2}$:
\[
\begin{array}{c}
\bar{d}_{approx, 2}(\bar{x}_1, \bar{x}_2)\\
= \min \left( \bar{d}_{approx, 1}(\bar{x}_1, \bar{x}_2), \min_{\bar{w}} dist(\bar{x}_1, \bar{x}_2, \bar{w}) \right),
\end{array}
\]
where $\bar{w}$ loops through the set of all trees generated from some tree $\bar{y}$ which is a collapsed version of both $\bar{x}_1$ and $\bar{x}_2$.
\end{itemize}

\begin{rem}[Paths with one Euclidean stretch]

The simplest possible path between two tree-shapes does not visit any identified subspaces, and has length $\bar{d}_{approx, 1}$. We have not in an obvious way taken into account that this might in fact be the shortest path. In order to see that this is part of the computation, let us start out as above:

\begin{itemize}
\item Let $\bar{x}_1$, $\bar{x}_2$ be tree-shapes in $\bar{X}$. Let $\bar{y}$ be some tree-shape whose topological structure can be obtained as a collapsed version of both $\bar{x}_1$ and $\bar{x}_2$. 
\item The topology of the tree $\bar{y}$ is obtained from $\bar{x}_1$ (and similarly from $\bar{x}_2$) by collapsing two types of edges:
\begin{itemize}
\item internal edges (that is, there is some edge in the tree which is further from the root in $\bar{x}_1$ (or $\bar{x}_2$), which is not being collapsed), or
\item external edges (that is, all the edges in $\bar{x}_1$ (or $\bar{x}_2$) which are farther from the root are also collapsed).
\end{itemize}
\item Assume that it is possible to find some tree $\bar{y}$ which is obtained by only collapsing external edges.
\item The shortest path passing through the corresponding $\bar{w}$ will be a single Euclidean stretch, which can be obtained as described above.
\item Hence, the option that $\bar{d}_{approx, 1}$ might give the optimal path is part of the computation of $\bar{d}_{approx, 2}$.
\end{itemize}

\end{rem}

\begin{rem}
There is a chance that the intermediate trees appearing in this algorithm are ''too large'' for $\bar{X}$ (for example as in fig.~\ref{algorithm}. We have chosen to disregard this, and cut the intermediate trees off in iterative algorithms if they are too large. This can also be controlled by cutting $\bar{w}$ off by default.
\end{rem}

We still have not addressed the question of where the collapsed tree $\bar{y}$ comes from.

\begin{itemize}
\item In order to compute $\bar{d}_{approx, 2}(\bar{x}_1, \bar{x}_2)$, we need to search through all different collapsed versions $\bar{y}$ of $\bar{x}_1$ and $\bar{x}_2$.
\item In order to reduce computation time, we approximate this by setting a bound $D$ on the degree of vertices of the collapsed versions. 
\item For $D = 3$, the number of such subtrees is $O(n)$, where $n$ is the number of edges in the tree. For $D > 3$, the number of such subtrees is $O(n^{(D-2)})$. In our experiments we use $D = 3$.
\end{itemize}

\subsection{Paths with several Euclidean stretches}

Note that if we can answer the question in vii) in the beginning of section~\ref{reformulation}, then we can compute $\bar{d}_{approx, K}$ as well:
\[
\begin{array}{c}
\bar{d}_{approx, K}(\bar{x}_1, \bar{x_2}) =\\
 \\
\min 
\left\{
\begin{array}{c}
  \bar{d}_{approx, K-1}(\bar{x}_1, \bar{x_2}),\\
\min_{\bar{w}} 
 \left( 
 \bar{d}_{approx, 2}(\bar{x}_1, \bar{w}) + \bar{d}_{approx, K-1}(\bar{x}_2, \bar{w}) 
  \right) 
\end{array}
\right\}
\end{array}
\]
where $\bar{w}$ goes through a family of trees derived from all possible collapsed trees $\bar{y}$ obtained from $\bar{x}_1$.

\begin{rem}
This approach can be adapted to computation in a more restricted space $\bar{Z}$ by making sure that all the intermediate trees are contained in $\bar{Z}$.
\end{rem}

\end{appendices}

\end{document}